\documentclass[a4paper,12pt]{article}
\usepackage{amsmath}
\usepackage{amsfonts}
\usepackage{times}
\usepackage{amsthm}

\usepackage{pstricks,pst-node,pst-coil}
\usepackage{graphics,geometry,epsfig}
\newcommand{\DATUM}{November 13, 2008}              
\newcommand{\comma}{\: ,}     
\newcommand{\period}{\: .}    
\newcommand{\Proof}{\noindent\emph{Proof. }}              
\newcommand{\QED}{\hspace*{\fill}\mbox{$\Box$}}           
\newcommand{\eps}{{\varepsilon}}        
\newcommand{\vphi}{{\varphi}}           
\newcommand{\Om}{\Omega}                
\newcommand{\om}{\omega}
\newcommand{\la}{\langle}
\newcommand{\ra}{\rangle}
\newcommand{\C}{C_\chi}
\newcommand{\one}{\mathbf{1}}

\newcommand{\cB}{\mathcal{B}}
\newcommand{\cD}{\mathcal{D}}

\newcommand{\cF}{\mathcal{F}}

\newcommand{\cH}{\mathcal{H}}

\newcommand{\cM}{\mathcal{M}}         
\newcommand{\cO}{\mathcal{O}}         

\newcommand{\cR}{\mathcal{R}}
\newcommand{\cS}{\mathcal{S}}
\newcommand{\cT}{\mathcal{T}}
\newcommand{\cU}{\mathcal{U}}

\newcommand{\cW}{\mathcal{W}}

\newcommand{\field}[1]{\mathbb{#1}}

\newcommand{\RR}{\field{R}}     
\newcommand{\NN}{\field{N}}     
\newcommand{\CC}{\field{C}}     
\newcommand{\fh}{\mathfrak{h}}  
\newcommand{\bchi}{{\overline{\chi}}}
\newcommand{\btau}{{\overline{\tau}}}

\newcommand{\uw}{{\underline w}}
\newcommand{\umpnq}{{\underline{m,p,n,q}}}
\newcommand{\upq}{{\underline{0,p,0,q}}}
\newcommand{\tuw}{{\underline{\tilde{w}}}}
\newcommand{\huw}{{\underline{\hat{w}}}}

\newcommand{\hE}{\widehat{E}}
\newcommand{\hT}{\widehat{T}}    

\newcommand{\hw}{\hat{w}}

\newcommand{\tV}{\widetilde{V}} 
\newcommand{\tW}{\widetilde{W}} 
\newcommand{\tk}{\tilde{k}}

\newcommand{\tr}{\tilde{r}}
\newcommand{\tw}{\widetilde{w}}
\newcommand{\tx}{\tilde{x}}

\newcommand{\tpi}{\tilde{\pi}}

\newcommand{\sym}{\mathrm{sym}}
\newcommand{\const}{\mathrm{const}}
\renewcommand{\red}{\mathrm{red}}
\newcommand{\rIm}{\mathrm{Im}}
\newcommand{\rRe}{\mathrm{Re}}               
\newcommand{\Ran}{\mathrm{Ran}}              

\newcommand{\cirS}{\mathop{\bigcirc\kern -.73em {\scriptstyle{\rm S}}}}

\newcommand{\dist}{{\rm dist}}

\newcommand{\dom}{\mathrm{dom}}
\newcommand{\supp}{\mathrm{supp}}

\newcommand{\cern}{\mathrm{Ker}}
\newcommand{\op}{\mathrm{op}}

\newcommand{\rad}{{f}}

\newcommand{\hf}{H_\rad}

\newcommand{\sprod}[2]{\mbox{$\langle #1,#2 \rangle$}}   
\newcommand{\lb}{\left(}
\newcommand{\rb}{\right)}
\renewcommand{\thesection}
{\Roman{section}}                      
\renewcommand{\theequation}
{\thesection.\arabic{equation}}        
\newcommand{\secct}[1]{\section{#1}
\setcounter{equation}{0}}              

\newtheorem{theorem}{Theorem}[section]         
\newtheorem{lemma}[theorem]{Lemma}             
\newtheorem{corollary}[theorem]{Corollary}     
\newtheorem{remark}[theorem]{Remark}           
\newtheorem{proposition}[theorem]{Proposition} 
\theoremstyle{plain}
\begin{document}
\bibliographystyle{plain}
\setcounter{page}{0}
\thispagestyle{empty}

\title{Spectral Renormalization Group}
\author{
J\"{u}rg Fr\"{o}hlich
\thanks{Institute~for~Theoretical Physics; ETH Z\"urich; Switzerland; and IHES, Bures-sur-Yvette, France}
\and Marcel Griesemer
\thanks{Department ~of Mathematics, University~of Stuttgart, D-70569 Stuttgart, Germany}
%
\and Israel Michael Sigal
\thanks{Department ~of Mathematics, University of Toronto; Toronto; Canada}
\thanks{Supported by NSERC Grant No. NA7901} \\
}

\date{\DATUM}

\maketitle

\begin{abstract}
The operator-theoretic  renormalization group (RG) methods are
powerful analytic tools to explore spectral properties of
field-theoretical models such as quantum electrodynamics (QED) with
non-relativistic matter. In this paper these methods are extended
and simplified. In a companion paper, our variant of
operator-theoretic RG methods is applied to establishing the
limiting absorption principle in non-relativistic QED near the
ground state energy.

%
%
%
%

\end{abstract}
%

\setcounter{page}{1}

\secct{Introduction} \label{sec-I}

This paper is devoted to the nuts and bolts of the spectral
(operator-theoretic) renormalization group (RG) method introduced in
\cite{BachFroehlichSigal1998a,BachFroehlichSigal1998b} and developed
further in \cite{BachChenFroehlichSigal2003,GriesemerHasler2}. This method has been
used successfully in order to describe the
spectral structure of non-relativistic quantum electrodynamics (QED)
with confining potentials and of Nelson's model with a 'subcritical'
interaction \cite{ BachFroehlichSigal1998a,BachFroehlichSigal1998b,
Chen2001, BachChenFroehlichSigal2006, Faupin2007, Sigal2008} (see
\cite{GustafsonSigal} for a book exposition and \cite{
BachFroehlichPizzo1, BachFroehlichPizzo2,
FroehlichGriesemerSigal2008a}, for an alternative multiscale
technique). The RG technique developed in this paper is a variant of
the one presented in \cite{BachChenFroehlichSigal2003}, where the smooth Feshbach-Schur map was introduced. It  simpler than that of \cite{BachChenFroehlichSigal2003} and similar to that of \cite{GriesemerHasler2}.
%

In this paper we apply the RG technique to prove existence of
eigenvaules and describe continuous spectra for operators on Fock
spaces appearing in massless quantum field theories for which
standard techniques do not work. (The papers
\cite{BachChenFroehlichSigal2003,GriesemerHasler2} deal only with
eigenvalues.) The results obtained here are used in subsequent
papers to prove existence of the ground state and resonances for
non-relativistic QED without the confinement assumption
(\cite{Sigal2008}, see also \cite{BachFroehlichPizzo1} ) and to
prove local decay near the ground state energy
(\cite{FroehlichGriesemerSigal2008b}, see also
\cite{FroehlichGriesemerSigal2008a}).



The class of Hamiltonians and the problems we consider here
originate in
non-relativistic QED. This theory deals with the interactions of
non-relativistic matter with the quantized electro-magnetic field.
(See \cite{Cohen-TannoudjiDupont-RocGrynberg1991,
Cohen-TannoudjiDupont-RocGrynberg1992, GustafsonSigal, Spohn} for
background.)

The dynamics of non-relativistic matter is generated by the
Schr\"odinger operator
\begin{equation} \label{Hp}
H_p:=-\sum\limits_{j=1}^n \frac{1}{2m_j} \Delta_{x_j}+V(x),
\end{equation}
where $\Delta_{x_j}$ is the Laplacian in the variable $x_j$,
$x=(x_1,\dots,x_n)$, and $V(x)$ is the potential energy of the
particle system.  This operator acts on the Hilbert space $\cH_{p}$,
which is either $L^2(\mathbb{R}^{3n})$ or a subspace of this space
determined by a symmetry group of the particle system. We assume
that $V(x)$ is real and s.t. the operator $H_p$ is self-adjoint.

The quantized electromagnetic field is described by the quantized
vector potential
\begin{equation}\label{A}
A(y)=\int(e^{iky}a(k)+e^{-iky}a^*(k))\chi(k)\frac{d^3k}{\sqrt{|k|}}
\end{equation}
in the Coulomb gauge ($div A(x) =0$). Here $\chi$ is an ultraviolet
cut-off: $\chi(k)=\frac{1}{(2\pi)^3 \sqrt{2}} $ in a neighborhood of
$k=0$, and $\chi$ vanishes rapidly at infinity.
The dynamics of the quantized electromagnetic field is given by the
quantum Hamiltonian
\begin{equation} \label{Hf}
\hf \ = \ \int d^3 k \om(k) \; a^*(k)\;  a(k) .
\end{equation}
The operators $A(y)$ and $H_p$ act on the Fock space $\cH_{f}\equiv
\cF$. Above, $\om(k) \ = \ |k|$ is the dispersion law connecting the
energy, $\omega(k)$, of the field quantum with its wave vector $k$,
and $a^*(k)$ and $a(k)$ denote the creation and annihilation
operators on $\cF$. The latter are operator-valued generalized,
transverse vector fields:
$$a^\#(k):= \sum_{\lambda \in \{0, 1\}}
e_{\lambda}(k) a^\#_{\lambda}(k),$$ where $e_{\lambda}(k)$ are
polarization vectors, i.e. orthonormal vectors in $\mathbb{R}^3$
satisfying $k \cdot e_{\lambda}(k) =0$, and $a^\#_{\lambda}(k)$ are
scalar creation and annihilation operators satisfying canonical
commutation relations. The right side of \eqref{Hf} can be
understood as a weak integral.  See the Supplement for a brief
review of definitions of the Fock space, the creation and
annihilation operators and  the operator $\hf$.

The Hamiltonian of the total system, matter and radiation field, is
given by
%
%
%
%
\begin{equation}\label{Hsm}
H^{}_g=\sum\limits_{j=1}^n  \frac{1}{2m_j}
(-i\nabla_{x_j}+gA(x_j))^2+V(x)+H_f
\end{equation}
acting on the Hilbert space $\cH:=\cH_{p}\otimes\cH_{f}$. Here
the coupling constant $g$ is related to the fine-structure constant
$\alpha =\frac{e^2}{4\pi \hbar c}\approx  \frac{1}{137}$. (See
\cite{BachFroehlichSigal1999, FroehlichGriesemerSigal2008a,
Sigal2008}for a discussion of the definition of $H^{}_g$ and units
involved.) This model describes
emission and absorption of radiation by systems of matter, such as
atoms and molecules, as well as other processes of interaction of
quantized radiation with matter. It has been extensively studied in
the last decade; see references in \cite{Sigal2008, Spohn} for
references to earlier contributions.

For a large class of potentials $V(x)$, including Coulomb
potentials, and for an ultra-violet cut-off in $A(x)$, the
operator $H^{}_g$ is self-adjoint.

%
The key problem of non-relativistic QED is to establish spectral and
resonance structure of $H^{}_g$ and, in particular, to prove
existence (and uniqueness) of the ground state and of resonances of
$H^{}_g$ corresponding to excited states of the atomic Hamiltonian.

One verifies that $\hf$ defines a positive, self-adjoint operator on
$\cF$ with purely absolutely continuous spectrum, except for a
simple eigenvalue $0$ corresponding to the vacuum eigenvector $\Om$
(see Supplement). Thus, for $g=0$, the low-energy spectrum of the
Hamiltonian $H^{}_0$ of the decoupled system consists of branches
$[\epsilon^{(p)}_i, \infty)$ of absolutely continuous spectrum,
where $\epsilon^{(p)}_i$ are the isolated eigenvalues of the
particle Hamiltonian $H_p$, and of the eigenvalues
$\epsilon^{(p)}_i$ sitting at the 'thresholds' of the continuous
spectrum. The absence of gaps between the eigenvalues and thresholds
is a consequence of the fact that the photons  are massless. This
leads to hard and subtle problems in perturbation theory, known
collectively as the infrared problem.

The first step in tackling the problem of ground states and
resonances in the framework of the RG approach is to perform a
certain canonical transformation and then apply to the resulting
Hamiltonian
a specially designed RG map in order to project out the particle-
and high-photon-energy degrees of freedom (\cite{Sigal2008} (cf.
\cite{BachFroehlichSigal1998a}).
As a result, one arrives at a Hamiltonian on Fock space of the form
$H := T + W$, where $T:=w_{0,0}[H_f]$, with $w_{0,0}: [0, \infty)
\rightarrow \mathbb{C}$ and continuous  ($w_{0,0}[H_f]$ is defined
by the operator calculus), and
%
%
%
\begin{eqnarray} \label{W}
W&:=&\sum_{m+n \geq 1}\chi_1 \int_{B^{m+n}_1} \prod_{1}^{m+n}(\frac{
dk_{j} }{ |k_{j}|^{1/2} }) \;  \prod_{1}^{m}a^*( k_{j} ) \,
\\ \nonumber
&& \times w_{m,n} \big[ \hf ; k_{1}, ...,k_{m+n}  \big] \,
\prod_{m+1}^{m+n}a( k_{j} ) \: \chi_1 ,
\end{eqnarray}
%
Here $w_{m,n}: I\times B_1^{m+n} \rightarrow \mathbb{C}, m+n
> 0$, $B_1^r$ denotes the Cartesian product of $r$ unit balls in $\RR^{3}$, $I:=[0,1]$ and
$\chi_1:=\chi_1(\hf)$ with $\chi_1(r)$ a smooth cut-off function
s.t. $\chi_1 = 1$ for $r \le 9/10,\ = 0$ for $r\ge 1$ and $0 \le
\chi_1(r) \le1\ $. See Section \ref{sec-III} for more details
concerning notation.
Operators on Fock space of the form above will be said to be in
\textit{generalized normal (or Wick) form}.

Note that, in order to be able to apply our theory to the analysis
of resonances of $H_g$, the operators $H=T+W$, introduced above, are
allowed to be non-self-adjoint.

Our goal in this paper is to describe the spectrum of the operator
$H$ near $0$. We assume that the function $w_{0,0}(r)$, defining the
operator $T:= w_{0,0}[H_f]$, satisfies
\begin{equation}\label{w00}
w_{0,0}(0) =0,\ \sup_{r \in [0,\infty)}| w'_{0,0}(r) - 1 | \leq
\beta_0.
\end{equation}
We consider the operator $W$ (see \eqref{W}) as a
perturbation of the operator $T:=w_{0,0}[H_f]$, whose spectrum is
explicitly known. It consists of the essential spectrum
$w_{0,0}(\overline{\mathbb{R}^+})$ and an eigenvalue $0$ at its tip
with the eigenvector $\Omega$. We propose to determine the effect of
the perturbation $W$ on the spectrum of $T$ near $0$ and, in
particular, to determine the fate of the eigenvalue $0$ of $T$. If
the operator $H$ has an eigenvalue near $0$,
we call it the ground state energy of $H$.

We denote by $\cD_{s}$ the set of operators of the form $H=T+W$,
where $T$ and $W$ are described above, such that \eqref{w00} holds
and $$ \| \uw_1 \|_{\mu,s, \xi} \leq \gamma_0,$$ where $\uw_1 :=
(w_{m,n})_{m+n \geq 1}$, and $ \| \uw_1 \|_{\mu,s, \xi}$ is a norm
defined in Section \ref{sec-III}. We define a subset $S$ of the
complex plane by
\begin{equation}
S:=\{w\in \mathbb{C}| \rRe w \ge 0, |\rIm
w| \le \frac{1}{3} \rRe w \}.
\end{equation}

Recall that a complex function $f$ on an open set $\cD$ in a
complex Banach space $\cB$ is said to be \textit{analytic} if
$\forall H\in\cD$ and $\forall \xi \in \cB,\ f(H+ \tau \xi)$ is
analytic in the complex variable $\tau$ for $|\tau|$ sufficiently
small (or equivalently, $f$ is G\^{a}teaux-differentiable, see \cite{Berger}; a stronger notion of analyticity, requiring in addition that $f$ is locally bounded, is used in \cite{HillePhillips}).  In the next theorem $\cB$ is the space 
of $H_f$-bounded operators on $\cF$ (i.e. the space of closed operators $A$ with $A(H_f+1)^{-1}$ bounded). We are now prepared to state the main
result of this paper.

\begin{theorem} \label{thm-main}
Assume that $\beta_0$ and $\gamma_0$ are sufficiently small. Then
there is
an analytic  map $e:\cD_s \rightarrow \mathbb{C}$ such that $e(H)
\in \mathbb{R}$, for $H=H^*$, and for $H\in \cD_s$ the number
$e(H)$ is
a simple eigenvalue of the operator $H$ and
$\sigma(H) \subset e(H) +S$.

\end{theorem}

Note that our approach also provides an effective way to compute the
eigenvalue $e(H)$ and the corresponding eigenvector.

Theorem \ref{thm-main} is used in \cite{Sigal2008,
FroehlichGriesemerSigal2008b}. Besides, our main technical result,
Theorem \ref{stable-manif} formulated in Section \ref{sec-V},
furnishes a key technical step in an RG proof of local decay, see
\cite{FroehlichGriesemerSigal2008b}.

Combining results of this paper with those of \cite{AFFS} one
obtains estimates on the resolvent of $H$ near the eigenvalue
$e(H)$: For each $\Psi$ and $\Phi$ from a dense set of vectors,
the matrix element $\langle \Psi, (H-z)^{-1}\Phi\rangle$  near the
eigenvalue $e \equiv e(H)$ of $H$ is of the form
\begin{equation} \label{poles}
\langle \Psi, (H-z)^{-1}\Phi\rangle =(e -z)^{-1} p(\Psi, \Phi) +
r(z, \Psi, \Phi) \comma
\end{equation}
%
where $p$ and $r(z)$ are sesquilinear forms in $\Psi$ and $\Phi$
with $r(z)$ analytic in $z \in Q:=  \mathbb{C}\backslash  (e(H) +S)
$ and bounded on the intersection of a neighbourhood of $e$ with $Q$
as
$$|r(z, \Psi, \Phi)| \le C_{\Psi, \Phi}|e-z|^{-\gamma}\
\mbox{for some}\ \gamma <1.$$
Such estimates are needed in an analysis of the long time dynamics
of resonances in QED; see \cite{AFFS}. This will be described in
more detail elsewhere.

Next, we explain the main ideas of the spectral renormalization
group method.
%
Our goal is to describe the spectral structure near $0$ of an
operator $H$ from the set $\cD_s$ introduced above. Denote by $D(0,
\alpha)$ the disc in $\mathbb{C}$ centered at $0$ and of radius
$\alpha$. For $\alpha_0$ sufficiently small,
%
we construct a renormalization transformation, $\cR_{\rho}$,
defined on $\cD: = D(0, \alpha_0)\mathbf{1}+ \cD_s$, with the
following properties:
\begin{itemize}
\item $\cR_{\rho}$ is 'isospectral' and 'preserves' the limiting
absorption principle;
\item $\cR_{\rho}$ removes the photon degrees of
freedom related to energies $\ge \rho$.
\end{itemize}

We then consider the discrete semi-flow, $\cR_{\rho}^n, n \ge 1$,
generated by the renormalization transformation, $\cR_{\rho}$
(called renormalization group) and relate the dynamics of this flow
to spectral properties of individual Hamiltonians in $\cD_s$. We
show that the flow, $\cR_{\rho}^n$, has the fixed-point manifold
$\cM_{fp}:=\mathbb{C}H_f$, an unstable manifold
$\cM_{u}:=\mathbb{C}\one$, and a (complex) co-dimension $1$ stable
manifold $\cM_s$ for $\cM_{fp}$ foliated by (complex) co-dimension
$2$ stable manifolds for each fixed point.  We show that
$H_{}-\lambda $ is in the domain of $\cR_{\rho}^n$, provided the
parameter $\lambda$ is adjusted appropriately, so that $H_{}-\lambda
$ is, roughly, in a $\rho^n-$neighborhood of the stable manifold
$\cM_s$.
\begin{center}
\psset{unit=1cm} \pspicture(-4,-5)(8,3.5)

\pscustom[linestyle=none,fillstyle=solid,fillcolor=lightgray]{
\psbezier(2,1)(3,0.5)(4,0)(4.5,-1.5)
\psbezier[liftpen=1](0,-4.5)(-0.25,-3)(-1,-2)(-2,-1)}
\psbezier[linewidth=0.5pt,linestyle=dashed](1,0.5)(2.2,-0.1)(3,-0.5)(3.5,-2.25)
\rput(0.5,-3.5){$\mathcal{M}_s$} \rput(0.5,0.8){$w H_f$}

\psset{linewidth=0.5pt} \psline(-3,-1.5)(3,1.5)
\rput(3.5,1.5){$\mathcal{M}_{fp}$}

\psline(0,-1)(0,2)\rput(0,2.3){$\mathcal{M}_u$}

\psline[linewidth=0.5pt,linestyle=dashed](1,0.5)(1,2.5)
\psbezier[linewidth=1pt]{->}(3,-0.7)(2,0.5)(1.2,0.5)(1.2,2)
\qdisk(1.4,0.97){2pt}\psline[linewidth=0.3pt]{<-}(1.5,1.1)(2.2,2)
\rput(3,2.2){$\mathcal{R}^{n}_{\rho}(H-\lambda)$}

\psline[linewidth=0.5pt,linestyle=dashed](3,0.2)(3,-0.7)
\qdisk(3,0.2){2pt}\qdisk(3,-0.7){2pt}\rput(3.4,0.2){$H$}
\psline[linewidth=0.3pt]{<-}(3.1,-0.7)(4.5,-0.5)\rput(5.2,-0.5){$H-\lambda$}
\endpspicture

Stable and unstable manifolds.
\end{center}
Thus, for $n$ sufficiently large, the operators
$H^{(n)}_\lambda:=\cR_{\rho}^n(H_{}-\lambda )$ are close to the
operator $wH_f$, for some $w\in \mathbb{C}$ with $Re\ w >0$,
and their spectra can be easily analyzed.
%
Since the renormalization map is 'isospectral', we can pass this
spectral information to
the operator $H^{(n-1)}_{\lambda }$, and so forth, until we obtain
the desired spectral information for the initial operator $H_{}$.


Our paper is organized as follows. In Section \ref{sec-II} we
describe the Feshbach-Schur map, which is the main ingredient of the
renormalization map introduced in Section \ref{sec-IV}. In Section
\ref{sec-III} we define the Banach spaces on which the
renormalization map acts. The renormalization group approach is
presented in Section \ref{sec-V} where the main technical results
implying Theorem \ref{thm-main} are proven. In Appendix I we present
the proof of a key technical result describing properties of the
renormalization map. This proof is close to the proof of a similar
result in \cite{BachChenFroehlichSigal2003} and is presented here
for the reader's convenience. In Appendix II we present a result
on the construction of eigenvalues and eigenvectors, similar to a
corresponding result of \cite{BachChenFroehlichSigal2003}. Finally,
in a Supplement, we collect some relevant facts on Fock space and
creation and annihilation operators.


\secct{The Smooth Feshbach-Schur Map} \label{sec-II}
%

In this section, we review  the method of isospectral decimation
maps acting on operators, introduced in
\cite{BachFroehlichSigal1998a,BachFroehlichSigal1998b} and refined
in \cite{BachChenFroehlichSigal2003}.
At the origin of this method is the \textit{isospectral smooth
Feshbach-Schur map}\footnote{In
\cite{BachFroehlichSigal1998a,BachFroehlichSigal1998b,
BachChenFroehlichSigal2003} this map is called the Feshbach map. As
was pointed out to us by F. Klopp and B. Simon, the invertibility
procedure at the heart of this map was introduced by I. Schur in
1917; it appeared implicitly in an independent work of H. Feshbach
on the theory of nuclear reactions, in 1958, where the problem of
perturbations of operator eigenvalues was considered. See
\cite{GriesemerHasler1} for further extensions and historical
remarks.} acting on a set of closed operators and mapping a given
operator to one  acting on a subspace of the original Hilbert space.


%
%
Let $\chi$,  $\bchi$ be a partition of unity on a separable Hilbert
space $\cH$, i.e. $\chi$ and  $\bchi$ are positive operators on
$\cH$ whose norms are bounded by one, $0 \leq \chi, \bchi \leq
\mathbf{1}$, and $\chi^{2}+ \bchi^{2} = \mathbf{1}$. We assume that
$\chi$ and $\bchi$ are nonzero. Let $\tau$ be a (linear) projection
acting on closed operators on $\cH$ with the property that operators
in its image commute with $\chi$ and $\bchi$. We also assume that
$\tau(\textbf{1}) =\textbf{1}$.
Let $\overline{\tau}:= \mathbf{1} - \tau$ and define
\begin{equation}
\\ \label{II-1}
H_{\tau,\chi^{\#}} \ \; :=  \tau(H) \: + \: \chi^{\#}
\overline{\tau}(H)\chi^{\#} \period
\end{equation}
where $\chi^{\#}$ stands for either $\chi$ or $\bchi$.

Given $\chi$ and $\tau$ as above, we denote by $D_{\tau,\chi}$ the
space of closed operators, $H$, on $\cH$ which belong to the domain
of $\tau$ and satisfy the following three conditions:

(i) $\tau$ and $\chi$ (and therefore also $\btau$ and $\bchi$) leave
the domain $D(H)$ of $H$ invariant:
\begin{equation}
\label{II-2} D(\tau(H))=D(H)\ \mbox{and}\  \chi D(H)\subset D(H),
\end{equation}

(ii)
\begin{equation}
\label{II-3}   H_{\tau,\bchi}\ \mbox{is (bounded) invertible on}\
\Ran \, \bchi,
\end{equation}
and

(iii)
\begin{equation}
\label{II-4}\overline{\tau}(H) \chi\ \mbox{and}\ \chi
\overline{\tau}(H)\ \mbox{extend to bounded operators on}\ \cH.
\end{equation}
(For more general conditions see \cite{BachChenFroehlichSigal2003,
GriesemerHasler1}.)

The \textit{smooth Feshbach-Schur map (SFM)} maps operators on $\cH$
belonging to $D_{\tau,\chi}$ to operators on $\cH$ by $H \ \mapsto \
F_{\tau,\chi} (H)$, where
\begin{equation} \label{II-5}
 F_{\tau,\chi} (H) \ := \ H_0 \, + \, \chi W\chi \, -
\, \chi W \bchi H_{\tau,\bchi}^{-1} \bchi W \chi \period
\end{equation}
Here $H_0 := \tau(H)$ and $W := \overline{\tau}(H)$. Note that $H_0$
and $W$ are closed operators on $\cH$ with coinciding domains, $
D(H_0)= D(W)=D(H)$, and $H = H_0 + W$. We remark that the domains of
$\chi W\chi$, $\bchi W\bchi$, $H_{\tau,\chi}$, and $H_{\tau,\bchi}$
all contain $D(H)$.

Remarks

\begin{itemize}
\item The definition of the smooth Feshbach map
given above differs somewhat from the one given in
\cite{BachChenFroehlichSigal2003}. In
\cite{BachChenFroehlichSigal2003}, the map $F_{\tau,\chi} (H)$ is
denoted by $F_{\chi}(H,\tau(H))$, and the pair of operators $(H, T)$
are referred to as a Feshbach pair.

\item The usual Feshbach-Schur map is obtained as a special case of the smooth Feshbach-Schur map by
choosing  $\chi=$ projection, and, usually, $\tau = 0$.

\item Typically the operator $\chi$ is taken to be of the form
$\chi := \chi(A)$ for some self-adjoint operator $A$ on $\cH$.
For the Feshbach map, $\chi$ has to be a projection and therefore we
would have to take $\chi := \chi(A)$ to be a characteristic function
of the operator $A$, while in the smooth Feshbach-Schur map we are
allowed to take $\chi := \chi(A)$ to be a smooth approximation of
the characteristic function of an interval in $\mathbb{R}$. This
explains the adjective 'smooth' in the definition.

\item In \cite{BachChenFroehlichSigal2003}  a semi-group property of $ F_{\tau,\chi} (H)$ is exhibited.
\end{itemize}
Next, we introduce some maps appearing in various identities
involving the Feshbach-Schur map:
\begin{eqnarray} \label{eq-II-4}
Q_{\tau,\chi} (H) & := & \chi \: - \: \bchi \, H_{\tau,\bchi}^{-1}
\bchi W \chi \comma
\\  \label{eq-II-5}
Q_{\tau,\chi} ^\#(H) & := & \chi \: - \: \chi W \bchi \,
H_{\tau,\bchi}^{-1} \bchi \period
\end{eqnarray}
Note that $Q_{\tau,\chi} (H) \in \cB( \Ran\, \chi , \cH)$ and
$Q_{\tau,\chi}^\#(H) \in \cB( \cH , \Ran\, \chi)$.

The smooth Feshbach-Schur map of $H$ is isospectral to $H$ in the
sense of the following theorem.
%
\begin{theorem} \label{thm-II-1}
Let $\chi$ and $\tau$ be as above, and assume that $H \in
D_{\tau,\chi}$  so that $F_{\tau,\chi} (H)$ is well defined. Then
\begin{itemize}
\item[(i)] $0 \in \rho(H) \Leftrightarrow 0 \in \rho(F_{\tau,\chi} (H))$, i.e. $H$ is bounded invertible on $\cH$ if and only if
$F_{\tau,\chi} (H)$ is bounded invertible on $\Ran\, \chi$.
\item[(ii)] If $\psi \in \cH \setminus \{0\}$ solves $H \psi = 0$
then $\vphi := \chi \psi \in \Ran\, \chi \setminus \{0\}$ solves
$F_{\tau,\chi} (H) \, \vphi = 0$. \item[(iii)] If $\vphi \in \Ran\,
\chi \setminus \{0\}$ solves $F_{\tau,\chi} (H) \, \vphi = 0$ then
$\psi := Q_{\tau,\chi} (H) \vphi \in \cH \setminus \{0\}$ solves $H
\psi = 0$. \item[(iv)] The multiplicity of the spectral value
$\{0\}$ is conserved under the Feshbach-Schur in the sense that
$\dim \cern H = \dim \cern F_{\tau,\chi} (H)$.
\item[(v)] If one of the inverses, $H^{-1}$ or $F_{\tau,\chi} (H)^{-1}$, exists then so does the
other, and these inverses are related by
\begin{equation} \label{eq-II-6}
H^{-1}  =  Q_{\tau,\chi} (H) \: F_{\tau,\chi} (H)^{-1} \:
Q_{\tau,\chi} (H)^\# \; + \; \bchi \, H_{\tau,\bchi}^{-1} \bchi .
\end{equation}
Moreover if $\tau(H)$ is invertible, then $$ F_{\tau,\chi} (H)^{-1}
= \chi \, H^{-1} \, \chi \; + \; \bchi \, \tau(H)^{-1} \bchi
\period$$
%
\end{itemize}
\end{theorem}

This theorem is proven in \cite{BachChenFroehlichSigal2003}; see
\cite{GriesemerHasler1} for further extensions.

In comparison with the original use of the Feshbach projection
method as a tool in the analytic perturbation theory of
eigenvalues, the smooth Feshbach-Schur map has two new features:
\begin{itemize}
\item Flexibility in the choice of the projection; in
particular, 'dressing' the eigenspace corresponding to some
eigenvalue with vectors from the continuous spectrum subspace, and
relaxing the projection property altogether;
\item Viewing the
Feshbach-Schur procedure as a map on a space of operators, rather
then a tool in the analysis of a \textit{fixed} operator.
Our operator theoretic renormalization group is based on an
iterative composition of Feshbach-Schur maps, decimating the degrees
of freedom of the system under investigation.
\end{itemize}
\secct{A Banach Space of Hamiltonians} \label{sec-III}
%
We construct a Banach space of Hamiltonians on which our
renormalization transformation will be defined. In order not to
complicate matters unnecessarily, we will think of the creation and
annihilation operators used below as scalar operators
neglecting helicity of photons. We explain at the end of the
Supplement how to reinterpret our expressions for the
photon creation and annihilation operators.

Recall that $B_1^r$ denotes the Cartesian product of $r$ unit balls
in $\RR^{3}$, $I:=[0,1]$ and $m,n \ge 0$. Given functions $w_{0,0}:
[0, \infty) \rightarrow \mathbb{C}$ and $w_{m,n}: I\times B^{m+n}
\rightarrow \mathbb{C}, m+n > 0$, we consider monomials, $W_{m,n}
\equiv W_{m,n}[w_{m,n}]$, in the creation and annihilation operators
defined as follows:
%
%
%
$W_{0,0}[w_{0,0}]:=w_{0,0}[H_f]$ (defined by the functional
calculus),
and
\begin{eqnarray} \label{III.1}
&&W_{m,n}[w_{m,n}]  :=
\\ \nonumber
&&  \int_{B_1^{m+n}} \frac{ dk_{(m,n)} }{ |k_{(m,n)}|^{1/2} } \;
a^*( k_{(m)} ) \, w_{m,n} \big[ \hf ; k_{(m,n)} \big] \, a(
\tk_{(n)} ) \:  ,
\end{eqnarray}
%
for $m+n>0$. Here we are using the notation
\begin{eqnarray} \label{III.2}
& k_{(m)} \: := \: (k_1, \ldots, k_m) \: \in \: \RR^{3m} \comma
\hspace{5mm}
a^*( k_{(m)} ) \: := \: \prod_{i=1}^m a^*(k_i ),
\\  \label{III.3}
& k_{(m,n)} \: := \: (k_{(m)}, \tk_{(n)}) \comma \hspace{5mm}
dk_{(m,n)} \: := \: \prod_{i=1}^m  d^3 k_i \; \prod_{i=1}^n d^3
\tk_i \comma &
\\  \label{III.4}
& |k_{(m,n)}| \, := \, |k_{(m)}| \cdot |\tk_{(n)}| \comma
\hspace{3mm} |k_{(m)}| \, := \, |k_1| \cdots |k_m| \period &
\end{eqnarray}
The notation $W_{m,n}[w_{m,n}]$ stresses the dependence of $W_{m,n}$
on $w_{m,n}$. Note that $W_{0,0}[w_{0,0}]$ $ := w_{0,0}[\hf]$. We
also denote $T\equiv W_{0,0}[w_{0,0}]$.

We assume that, for every $m$ and $n$ with $m+n>0$, the function
$w_{m,n}[ r; , k_{(m,n)}]$
is measurable in $k_{(m,n)} \in B_1^{m+n}$ 
and $s$ times continuously differentiable in $r \in I$, for some $s
\ge 1$, and for almost every $k_{(m,n)} \in B_1^{m+n}$. 
As a function of $k_{(m,n)}$, it is totally symmetric w.~r.~t.\
the variables $k_{(m)} = (k_1, \ldots, k_m)$ and $\tk_{(n)} =
(\tk_1, \ldots, \tk_n)$ and obeys the norm bound
\begin{equation} \label{III.5}
\| w_{m,n} \|_{\mu,s} \ :=
\sum_{n=0}^{s} \|  \partial_r^n w_{m,n} \|_{\mu} \ < \ \infty
\comma
\end{equation}
where
%
%
%
\begin{equation} \label{III.6}
\| w_{m,n} \|_{\mu} \ := \max_j \sup_{r \in I, k_{(m,n)} \in
B_1^{m+n}} \big| | k_j|^{-\mu}w_{m,n}[r ; k_{(m,n)}] \big|
\end{equation}
for some $\mu \ge 0$.  Here and in what follows, $k_j$ is one of the
$3-$vectors in the variable $k_{(m,n)}$. Recall that
$|k_{(m,n)}|^{-1/2}$ is absorbed in the integration measure in
the definition of $W_{m,n}$. For $m+n=0$ the variable $r$ ranges
over $[0,\infty)$, and we assume that the following norm is finite:
\begin{equation}
\\ \label{III.7}
\ \| w_{0,0} \|_{\mu, s} := |w_{0,0}(0)|+ \sum_{1 \le n \leq s}
\sup_{r \in [0,\infty)}|
\partial_r^n w_{0,0}(r)|.
 \hspace{10mm}
\end{equation}
(This norm is independent of $\mu$, but we keep this index for
notational convenience.) The Banach space of functions $w_{m,n}$ of
this type is denoted by $\cW_{m,n}^{\mu,s}$.
%
%
%
%
%

We fix three numbers $\mu$, $0 < \xi < 1$ and $s \ge 0$ and define
the Banach space
\begin{equation} \label{III.8}
\cW^{\mu,s} \ \equiv \cW^{\mu,s}_{\xi} := \ \bigoplus_{m+n \geq 0}
\cW_{m,n}^{\mu,s} \ \comma
\end{equation}
with the norm
\begin{equation} \label{III.9}
\big\|  \uw \big\|_{\mu, s,\xi} \ := \ \sum_{m+n \geq 0}
\xi^{-(m+n)} \; \| w_{m,n} \|_{\mu, s} \ < \ \infty \period
\end{equation}
Clearly, $\cW^{\mu',s'}_{\xi'} \subset \cW^{\mu,s}_{\xi}$ if $\mu'
\ge \mu, s' \ge s$ and $\xi' \le \xi$.

Let $\chi_1(r) \equiv\chi_{r\le1}$ be a smooth cut-off function s.t.
$\chi_1 = 1$ for $r \le 9/10,\ = 0$ for $r\ge 1$ and $0 \le
\chi_1(r) \le1\ $  and $\sup|\partial^n_r \chi_1(r)| \le 30\ \forall
r$ and for $n=1,2.$ We define $\chi_\rho(r) \equiv\chi_{r\le\rho}:=
\chi_1(r/\rho) \equiv\chi_{r/\rho\le1}$ and
$\chi_\rho\equiv\chi_{H_f\le\rho}$.
The following basic bound, proven in
\cite{BachChenFroehlichSigal2003}, links the norm defined in
(\ref{III.6})
to the operator norm on $\cB[\cF]$.
%
\begin{theorem} \label{thm-III.1}
Fix
$m,n \in \NN_0$ such that $m+n \geq 1$. Suppose that $w_{m,n} \in
\cW_{m,n}^{\mu,s}$, and let $W_{m,n} \equiv W_{m,n}[w_{m,n}]$ be as
defined in (\ref{III.1}). Then for all $\lambda >0$
%
\begin{equation} \label{III.10}
\big\|  (\hf+\lambda)^{-m/2} \, W_{m,n} \, (\hf+\lambda)^{-n/2}
\big\| \ \leq \  \| w_{m,n} \|_{0} \, ,
\end{equation}
and therefore
\begin{equation} \label{III.11}
\big\| \chi_\rho \, W_{m,n} \, \chi_\rho
 \big\|
\ \leq \ \frac{\rho^{(m+n)(1+\mu)}}{\sqrt{m! \, n!} } \, \| w_{m,n}
\|_{0} \, ,
\end{equation}
where $\| \, \cdot \, \|
$ denotes the operator norm on $\cB[\cF]$.
\end{theorem}

Theorem~\ref{thm-III.1} says that the finiteness of $\| w_{m,n}
\|_{0}$ insures that 
$\chi_1W_{m,n}\chi_1$ defines a bounded operator on $\cB[\cF]$.

With a sequence $\uw := (w_{m,n})_{m+n \geq 0}$ in $\cW^{\mu,s}$ we
associate an operator
by setting
\begin{equation} \label{III.12}
H(\uw)  := W_{0,0}[\uw] + \sum_{m+n \geq 1}
\chi_1W_{m,n}[\uw]\chi_1,
\end{equation}
where we write $W_{m,n}[\uw] := W_{m,n}[w_{m,n}]$. These operators
are said to be in \textit{generalized normal (or Wick) form} and are
called generalized Wick-ordered operators. Theorem~\ref{thm-III.1}
shows that the series in (\ref{III.12}) converges in the operator
norm and obeys the
estimate
\begin{equation} \label{eq-III-1-25.1}
\big\| \, H(\uw)- W_{0,0}(\uw) \, \big\| \ \leq \ \xi\big\| \, \uw_1
\, \big\|_{\mu,0, \xi} \comma
\end{equation}
for arbitrary $\uw = (w_{m,n})_{m+n \geq 0} \in \cW^{\mu,0}$ and any
$\mu
> -1/2$. Here $\uw_1 = (w_{m,n})_{m+n \geq 1}$. Hence we have the
linear map
\begin{equation} \label{eq-III-1-24.1}
H : \uw \to H(\uw)
\end{equation}
from $\cW^{\mu,0}$ into the set of closed operators on Fock space
$\cF$.
The following result is proven in \cite{BachChenFroehlichSigal2003}.
%
\begin{theorem} \label{thm-III-1-2}
For any $\mu \ge 0$ and $0 < \xi < 1$, the map $H : \uw \to H(\uw)$,
given in (\ref{III.12}), is injective.
%
%
\end{theorem}

Next, we decompose the Banach space $\cW^{\mu,s}$ into components
having, as we will establish below, distinct scaling properties. We
define the Banach spaces
\begin{equation} \label{eq-III-1-19}
\cT  := \Big\{ f \in \cW_{0,0}^{\mu,s}
\Big| \ f(0) = 0 \Big\}
\end{equation}
and
\begin{equation} \label{eq-III-1-17}
\cW_{1}^{\mu,s} \ := \ \bigoplus_{m+n \geq 1} \cW_{m,n}^{\mu,s}
\comma
\end{equation}
to consist of all sequences $\uw_1 := (w_{m,n})_{m+n \geq 1}$
obeying
\begin{equation} \label{III.17}
\| \uw_1 \|_{\mu, s,\xi} \ := \ \sum_{m+n \geq 1} \xi^{-(m+n)} \; \|
w_{m,n} \|_{\mu,s}\ < \ \infty \period
\end{equation}
We observe that there is a natural bijection
\begin{equation*}
\cW_{0,0}^{\mu,s} \ \to \ \CC \oplus \cT \comma\ w_{0,0} \ \mapsto \
w_{0,0}[0] \oplus (w_{0,0} - w_{0,0}[0]) \period
\end{equation*}
We shall henceforth not distinguish between $\cW_{0,0}^{\mu,s}$ and
$\CC \, \oplus \cT$. We rewrite our Banach $\cW^{\mu,s}$ space as
\begin{equation} \label{eq-III-1-22b}
\cW^{\mu,s} \  = \ \CC \; \oplus \; \cT \; \oplus \;
\cW_{1}^{\mu,s}.
\end{equation}
We define the spaces $\cW_{op}^{\mu,s} :=H(\cW^{\mu,s})$,
$\cW_{1,op}^{\mu,s} :=H(\cW_1^{\mu,s})$ and $\cW_{mn,op}^{\mu,s}
:=H(\cW_{mn}^{\mu,s})$. Sometimes we display the parameter $\xi$, as
in $\cW_{op,\xi}^{\mu,s} :=H(\cW^{\mu,s}_\xi)$. Theorem
\ref{thm-III-1-2} implies that $H(\cW^{\mu,s})$  is a Banach space
with norm $\big\| \, H(\uw) \big\|_{\mu,s, \xi}$ $:=\ \big\| \, \uw
\, \big\|_{\mu,s, \xi}$.

%
Corresponding to \eqref{eq-III-1-22b},  operators in
$\cW^{\mu,s}_{op}$ can be represented as
\begin{equation} \label{Hsplit} H(\uw)=E\one + T + W,
\end{equation}
where $E \in \CC$ is a complex number, $T = T[\hf]$, with $T[.]\in
\cT$,
and $W \in \cW^{\mu,s}_1$.
Indeed, let
\begin{equation} \label{eq-III-1-22a}
E:=w_{0,0}[0], T:=w_{0,0}[\hf] - w_{0,0}[0]\ \mbox{and}\
W:=\sum_{m+n \geq 1} \chi_1W_{m,n}[\uw]\chi_1. 
\end{equation}
Then the equation \eqref{Hsplit} holds.

\begin{remark} \label{rem-III.3} In this paper we need only $s=1$.
We introduce the more general spaces for the sake of future
references. Indeed, in our proof the limiting absorption principle
(LAP) in \cite{FroehlichGriesemerSigal2008b} we need $s=2$. More
precisely, we have to use more sophisticated Banach spaces where the
operator $\partial_r^n$ in \eqref{III.5}, is replaced by the
operator $\partial_r^n (k\partial_k)^q$ to \eqref{VIII.15}. Here
$q:= (q_1, \ldots, q_{M+N}),$ $ (k\partial_k)^q: = \prod_1^{M+N}(k_j
\cdot \nabla_{k_j})^{q_j}$, with $k_{m+j} := \tk_j$, and
the indices $n$ and $q$ satisfy $0 \le n+|q| \leq s$ with $s=2$.
\end{remark}

%
\secct{The Renormalization Transformation $\cR_\rho$} \label{sec-IV}
In this section we introduce an operator-theoretic renormalization
transformation based on the smooth Feshbach-Schur map, which is
closely related to the one introduced in
\cite{BachChenFroehlichSigal2003} and
\cite{BachFroehlichSigal1998a,BachFroehlichSigal1998b}. We fix the
index $\mu$ in our Banach spaces at some positive value, $\mu > 0$.

The renormalization transformation is homothetic to an isospectral
map defined on a polydisc in a suitable Banach space of
Hamiltonians. It has a certain contraction property insuring that
(upon appropriate tuning of the spectral parameter) the image of
any Hamiltonian in the polydisc under a large number of iterations
of the renormalization transformation approaches a fixed-point
Hamiltonian, $wH_f$, whose spectral analysis is particularly
simple. Thanks to the isospectrality of the renormalization map,
certain properties of the spectrum of the initial Hamiltonian can
be derived from the corresponding properties of the limiting
Hamiltonian.

The renormalization map is defined below as  a composition of a
decimation map, $F_{\rho}$, and two rescaling maps, $S_\rho$ and
$A_\rho$. Here $\rho$ is a positive parameter - the photon energy
scale - which will be chosen later.


The \emph{decimation of degrees of freedom} is accomplished by the
smooth Feshbach map, $F_{\tau,\chi}$ with
the operators $\tau$ and $\chi$  chosen as
\begin{equation}
\tau(H)=W_{00}:=w_{0 0}(H_f)\ \mbox{and}\
\chi=\chi_\rho\equiv\chi_{H_f\le \rho} , \label{IV.1}
\end{equation}
where $H=H(\uw)$ is given in  Eqn \eqref{III.12}. With $\tau$ and
$\chi$ identified in this way we will use the notation
\begin{equation}
F_\rho\equiv F_{\tau,\chi_\rho}. \label{IV.2}
\end{equation}
The decimation map acts on the Banach space $\cW_{op}^s$.

Let $\overline{\chi}_\rho$ be defined so that
$\chi_\rho\equiv\chi_{H_f\le\rho}\ \mbox{and}\
\overline{\chi}_\rho\equiv\chi_{H_f\ge\rho}$
form a smooth partition of unity,
$\chi_\rho^2+\overline{\chi}_\rho^2=\one$. The lemma below shows
that the domain of this map contains the following polydisc in
$\cW_{op}^{\mu,s}$:
\begin{eqnarray} \label{disc}
\cD^{\mu,s}(\alpha,\beta,\gamma)  && :=  \Big\{ H(\uw) \in
\cW_{op}^{\mu,s} \ \Big| \ |E|\leq\alpha \comma \\
 &&\sup_{r \in [0,\infty)}| T'[r] - 1 | \leq
\beta,\  \| \uw_1 \|_{\mu,s, \xi}\leq\gamma \Big\},\nonumber
\end{eqnarray}
for appropriate $\alpha, \beta, \gamma >0$. Here $H(\uw)=E+T+W$,
where $E$, $T$ and $W$ are given in \eqref{eq-III-1-22a} and $\uw_1
:=(w_{m,n})_{m+n \geq 1}$.

\begin{lemma} \label{lem-III-2-2}
Fix $0 < \rho < 1$, $\mu > 0, s \geq 1$, and $0 < \xi < 1$. Then
it follows that the polydisc $\cD^{\mu,s}(\rho/8, 1/8, \rho/8)$ is
in the domain of the Feshbach map $F_\rho$.
\end{lemma}
\Proof Let  $H(\uw) \in \cD^{\mu,s}(\rho/8, 1/8, \rho/8)$. We remark
that $W:= H(\uw)-E-T$ defines a bounded operator on $\cF$, and we
only need to check the invertibility of $H(\uw)_{\tau \chi_\rho}$ on
$\Ran \,\bchi_\rho$. Now the operator $E+T = W_{0,0}[\uw]$ is
invertible on $\Ran \,\bchi_\rho$ since
for all $r \in [3\rho/4, \infty)$
\begin{eqnarray} \label{IV.4}
Re\ T[ r] + Re\ E & \geq & r \, - \, | T[ r] - r | \, - \, |E|
\nonumber \\ & \geq & r \big( 1 \, - \, \sup_{r} | T'[
r] - 1 | \big) \: - \:  |E| \nonumber \\
& \geq & \frac{3 \, \rho}{4} ( 1 - 1/8 ) \: - \: \frac{\rho}{8} \
\geq \ \frac{ \rho}{2} \
\end{eqnarray}
and $T:=T[ \hf]$. Eqn \eqref{IV.4} implies also that $\|(E+T)^{-1}\|
\le 2/\rho$. On the other hand, by \eqref{III.11}, $\big\| W \|\leq
\xi\rho/8 \leq \rho/8$. Hence $\big\|\bchi_\rho W \bchi_\rho (E +
T)^{-1}\|\leq 1/4$ and therefore
$H(\uw)_{\tau, \bchi_\rho}= [1+\bchi_\rho W \bchi_\rho (E +
T)^{-1}](E + T)$ is invertible on $\Ran \,\bchi_\rho$. \QED

The last part of the proof above gives the estimate
\begin{eqnarray} \label{IV.4a}
\|(H(\uw)_{\tau \chi_\rho})^{-1}\| \le \frac{8}{3\rho}.
\end{eqnarray}

We introduce the \textit{scaling transformation} $S_\rho: \cB[\cF]
\to \cB[\cF]$, by
%
%
%
\begin{equation} \label{IV.5}
S_\rho(\one) \ := \ \one \comma \hspace{5mm} S_\rho (a^\#(k)) := \
\rho^{-3/2} \, a^\#( \rho^{-1} k) \comma
\end{equation}
where $a^\#(k)$ is either $a(k)$ or $a^*(k)$ and $k \in \RR^3$.
%
On the domain of the decimation map $F_\rho$ we define the
renormalization map $\cR_\rho$ as
\begin{equation} \label{IV.6}
\cR_\rho:=  
\rho^{-1}S_\rho\circ F_\rho.
\end{equation}

\begin{remark} \label{remIV-2} The renormalization map above is different from the one
defined in \cite{BachChenFroehlichSigal2003}.
The map in \cite{BachChenFroehlichSigal2003}  contains an additional
change of the spectral parameter $\lambda:= -\la H\ra_\Omega$.
\end{remark}

We mention here some properties of the scaling transformation. It is
easy to check that $S_\rho (\hf) = \rho \hf$, and hence
\begin{equation} \label{eq-III-2-3}
S_\rho ( \chi_\rho) = \ \chi_1 \hspace{5mm} \mbox{and} \hspace{6mm}
\rho^{-1} S_\rho \big( \hf \big) \ = \  \hf \comma
\end{equation}
%
%
%
which means that the operator $\hf$ is a \emph{fixed point} of
$\rho^{-1} S_\rho$. Further note that $E \cdot \one$ \emph{is
expanded} under the scaling map, $\rho^{-1}  S_\rho(E \cdot \one) =
\rho^{-1} E \cdot \one$, at a rate $\rho^{-1}$. (To control this
expansion it is necessary to suitably restrict the spectral
parameter.)

Next, we show that the interaction $W$ contracts under the scaling
transformation. To this end we remark that the scaling map $S_\rho$
restricted to $\cW_{op}^{\mu,s}$
induces a scaling map $s_\rho$ on $\cW^{\mu,s}$ by
\begin{equation} \label{eq-III-2-5}
\rho^{-1}  S_\rho \big( H(\uw) \big) \ =: \ H \big( s_\rho(\uw)
\big).
\end{equation}
It is easy to verify that $s_\rho(\uw):=(s_\rho(w_{m,n}) )_{m+n \geq
0} $ and,  for all $(m,n) \in \NN_0^2$,
\begin{equation} \label{eq-III-2-6}
s_\rho(w_{m,n}) \big[ r , k_{(m,n)} \big] \ = \ \rho^{m+n - 1} \:
w_{m,n}\big[ \rho \, r \; , \; \rho \, k_{(m,n)} \big] \period
\end{equation}
%
%
%
We note that by Theorem~\ref{thm-III.1}, the operator norm of
$W_{m,n} \big[ s_\rho(w_{m,n}) \big]$ is controlled by the norm

\begin{eqnarray*}
 \| s_\rho(w_{m,n}) \|_{\mu} & =& \max_j \sup_{r \in I, k \in B_1^{m+n}} \ \rho^{m+n - 1} \:
\frac{\big|w_{m,n}[\rho \, r \; , \; \rho \, k_{(m,n)}] \big|}{|
k_j|^{\mu}}\\
& \leq &
\rho^{m+n +\mu- 1} \, \| w_{m,n} \|_{\mu}.
\end{eqnarray*}
Hence, for $m+n \geq 1$, we have that
\begin{equation} \label{eq-III-2-8}
 \| s_\rho(w_{m,n}) \|_{\mu}  \leq \ \; \rho^{\mu}
\, \| w_{m,n} \|_{\mu}
\end{equation}
Since $\mu >0$, this estimate shows that $S_\rho$ contracts $\|
w_{m,n} \|_{\mu}$ by at least a factor of $\rho^{\mu} < 1$. The next
result shows that this contraction is actually a property of the
renormalization map $\cR_\rho$ along the 'stable' directions.
Recall, $\chi_{1}$ is the cut-off function introduced at the
beginning of Section III. Define the constant
\begin{equation} \label{Cchi}
C_\chi:=\frac{4}{3}\big(\sum_{n=0}^s \sup |
\partial_r^n \chi_1| + \sup |\partial_r \chi_1 |^2 \big) \le 200.
\end{equation}
Clearly, for, say, $s=1,\ C_\chi \ge 4/3$. We keep the constant
$C_\chi$ below in order to relate the analysis of this paper to that
of \cite{BachChenFroehlichSigal2003}.
%
\begin{theorem}
\label{thm-III-2-5} Let $\epsilon_0:H\rightarrow \la H\ra_\Omega$
and $\mu>0$.
Then for the absolute constant $C_\chi$ given in \eqref{Cchi} and
%
for any $s \ge 1,\ 0<\rho<1/2,\ \alpha,\beta \le \frac{\rho}{8}$ and
$\gamma \le \frac{\rho}{8C_\chi}$ we have that
\begin{equation}
\cR_\rho-\rho^{-1}\epsilon_0:\cD^{\mu,s}(\alpha,\beta,\gamma)\rightarrow
\cD^{\mu,s}(\alpha',\beta',\gamma'), \label{eqn:23}
\end{equation}
continuously, with $\xi:=\frac{\sqrt{\rho}}{4C_\chi}$ (in the
definition of the polydiscs, see \eqref{disc}) and
\begin{equation}
\alpha'=3C_\chi\lb\gamma^2/2\rho\rb,
\beta'=\beta+3C_\chi\lb\gamma^2/2\rho\rb, \gamma'=256
C_\chi^2\rho^\mu\gamma . \label{eqn:24}
\end{equation}
\end{theorem}
With some modifications, this theorem follows from Theorem 3.8 in
\cite{BachChenFroehlichSigal2003} and its proof; especially
Equations (3.104), (3.107) and (3.109). For the sake of
completeness, we present a proof of this theorem in Appendix I.

\begin{remark} \label{remIV-4} Subtracting the term $\rho^{-1}\epsilon_0$ from $\cR_\rho$
allows us to control the expanding direction during the iteration of
the map $\cR_\rho$. In \cite{BachChenFroehlichSigal2003} such
control was achieved by using a change of the spectral parameter
$\lambda$, which controls $\la H\ra_\Omega$.
%
\end{remark}

\secct{Renormalization Group} \label{sec-V}
%
In this section we describe some dynamical properties of iterations,
$\cR_\rho^n\ \forall n \ge 1 $, of the renormalization map
$\cR_\rho$.
A closely related iteration scheme is used in
\cite{BachChenFroehlichSigal2003}. First, we observe that
$$\forall \tau \in \mathbb{C},\ \cR_\rho(\tau H_f) = \tau H_f\
\mbox{and}\ \cR_\rho(\tau \textbf{1}) =\frac{1}{\rho} \tau
\textbf{1}.$$ Hence we define $\cM_{fp}:=\mathbb{C}H_f$  and
$\cM_{u}:=\mathbb{C}\textbf{1}$ as candidates for the manifold of
fixed points of $\cR_\rho$ and the unstable manifold.
The next result identifies the stable manifold of $\cM_{fp}$ which
turns out to be of (complex) codimension $\one$ and is foliated by
(complex) co-dimension $2$ stable manifolds, for each fixed point in
$\cM_{fp}$. This implies, in particular, that, in a vicinity of
$\cM_{fp}$, there are no other fixed points, and that $\cM_{u}$ is
the entire unstable manifold of $\cM_{fp}$ (see the figure on page
5).

We introduce some definitions. Recall that $D(\lambda,r):=\{z \in
\mathbb{C} | |z-\lambda| \le r \}$, a disc in the complex plane.
As an initial set of operators we take
$$\cD:=\cD^{\mu,s'}(\alpha_0,\beta_0,\gamma_0),$$ with $\alpha_0,
\beta_0,\gamma_0 \ll 1$ and $ s' \ge 1$. We also let
$$\cD_s:=\cD^{\mu,s'}(0,\beta_0,\gamma_0).$$ (The subindex $s$ stands for
'stable', not to be confused with the smoothness index $s$, which,
in this section, is denoted $s'$.) For $H \in \cD$ we write
$$H_u:=\la H\ra_\Omega\ \mbox{and}\ H_s:=H - \la H\ra_\Omega\
\mathbf{1}$$ (the unstable- and stable-central-space components of
$H$, respectively). Note that $H_s \in \cD_s$.

We fix the scale $\rho$ so that
\begin{equation}
\alpha_0, \beta_0,\gamma_0 \ll\rho\le 1/2.
\label{rho}
\end{equation}
Below, we use the $n-$th iteration of the numbers $\alpha_0,
\beta_0$ and $\gamma_0$ under the map \eqref{eqn:24}:
$$\alpha_n:=c\rho^{-1}(c\rho^\mu)^{2(n-1)}\gamma_0^2,$$
\begin{equation*}
\beta_n=\beta_0+\frac{c\gamma_0^2}{\rho}\sum_{j=0}^{n-1}(c\rho^\mu)^{2j},
\end{equation*}
$$\gamma_n=(c\rho^\mu)^n\gamma_0.$$

Recall that a vector-function $f$ from an open set $\cD$ in a
complex Banach space $\cB_1$ into a complex Banach space $\cB_2$ is
said to be \textit{analytic} iff $\forall H \in \cD$ and $\ \forall \xi \in \cB_1,\ f(H+ \tau \xi)$ is analytic in the complex variable $\tau$ for $|\tau|$ sufficiently small (see \cite{Berger}). One can show that $f$ is analytic iff it is
G\^{a}teaux-differentiable (\cite{Berger,
HillePhillips}). A stronger notion of analyticity, requiring in addition that $f$ is locally bounded, is used in
\cite{HillePhillips}. Furthermore, if $f$ is analytic in $\cD$ and $g$ is an analytic vector-function from an open
set $\Omega$ in
$\mathbb{C}$ into  $\cD$, then the composite function $f\circ g$ is
analytic on $\Omega$. In what follows $\cB_1$ is the space 
of $H_f$-bounded operators on $\cF$ and $\cB_2$ is either
$\mathbb{C}$ or $\cB (\cF)$.
%

For a Banach space $X$ the symbol $O_X(\alpha)$ will stand for an
element of $X$ bounded in its norm by $\textrm{const}\ \alpha$.

%
%

\begin{theorem} \label{stable-manif} Let $\delta_n :=\nu_n\rho^{n}$ with
$4 \alpha_n \leq \nu_{n} \leq\frac{1}{18}$. There is
an analytic  map $e:\cD_s \rightarrow D(0, 4 \alpha_0)$ s.t. $e(H)
\in \mathbb{R}$ for $H=H^*$, and
\begin{equation}
U_{\delta_n} \subset D(\cR_\rho^{n})\ \mbox{and}\
\cR_\rho^{n}(U_{\delta_n}) \subset
\cD^{\mu,s'}(\rho/8,\beta_{n},\gamma_{n}) \label{eqn:30aaa}
\end{equation}
where $U_\delta:= \{H \in \cD|\ |e(H_s)+ H_u|  \leq \delta\ \}.$
Moreover, $\forall H \in U_{\delta_n}$ and $\forall n \geq 1$, there
are $E_{n} \in \mathbb{C}$ and $\tau_n(r)\in \mathbb{C}$ s.t.
$|E_{n}| \leq 2\nu_n$, $|\ \tau_n(r)-1| \leq \beta_n$,  $ \tau_n$ is
$C^{s'}$,
\begin{equation}
\cR_\rho^n(H)=E_{n}+\tau_n(H_f)H_f +O_{\cW_{op}^{\mu,s'}}(\gamma_n),
\label{eqn:30b}
\end{equation}
(the spaces $\cW_{op}^{\mu,s'}$ are defined in Section \ref{sec-III}), $E_{n}$ and $\tau_n(r)$ are real if $H$
is self-adjoint and, as $ n
\rightarrow \infty$,  $\tau_n(r) $ converge in $L^\infty$ to some
number (constant function) $\tau \in \mathbb{C}$.
\end{theorem}
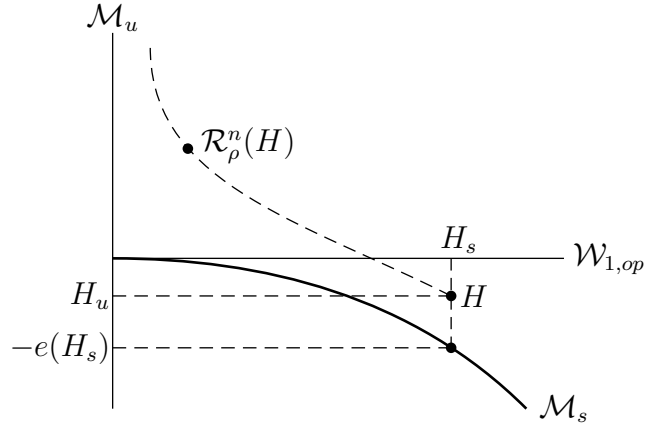
\begin{center}
\psset{unit=1cm}
\begin{pspicture}(-1,-3)(8,3.5)

\psset{linewidth=0.5pt} \psline(0,-2)(0,3)
\psline(0,0)(6,0) 

\psbezier[linewidth=1pt](0,0)(2,0)(4,-0.5)(5.5,-2) 
\psbezier[linewidth=0.5pt,linestyle=dashed](4.5,-0.5)(2.5,0.5)(0.5,1)(0.5,2.8)
\qdisk(1,1.46){2pt}

\qdisk(4.5,-0.5){2pt} \qdisk(4.5,-1.19){2pt}
\psline[linewidth=0.5pt,linestyle=dashed](4.5,0)(4.5,-1.2) 
\psline[linewidth=0.5pt,linestyle=dashed](0,-0.5)(4.5,-0.5)
\psline[linewidth=0.5pt,linestyle=dashed](0,-1.19)(4.5,-1.19)
\rput(-0.7,-1.19){$-e(H_s)$}
\rput(1.8,1.5){$\mathcal{R}^{n}_{\rho}(H)$}

\rput(-0.3,-0.5){$H_u$} \rput(4.8,-0.5){$H$}\rput(4.6,0.25){$H_s$}

\rput(6,-2){$\mathcal{M}_s$} \rput(0,3.2){$\mathcal{M}_u$}

\rput(6.6,0){$\mathcal{W}_{1,op}$}
\end{pspicture}
\end{center}

This theorem implies that $\cM_{fp}:=\mathbb{C}H_f$  is (locally) a
manifold of fixed points of $\cR_\rho$ and
$\cM_{u}:=\mathbb{C}\textbf{1}$ is the unstable manifold,
and the set
\begin{equation} \cM_s:=\bigcap_n U_{\delta_n}= \{H\in \cD |\ e(H_s)=-H_u\}\label{eqn:30a}
\end{equation}
is a local stable manifold for the fixed point manifold $\cM_{fp}$
in the sense that, $\forall H \in \cM_s,\ \exists \tau \in
\mathbb{C}$ s.t.
\begin{equation}\cR_\rho^n(H)\rightarrow \tau H_f\ \mbox{in the norm of}\
 \cW_{op}^{\mu,s'},  \label{eqn:30aa}
\end{equation}
as $ n \rightarrow \infty$. Moreover, $\cM_s$ is an invariant
manifold for $\cR_\rho$: $\cM_s \subset D(\cR_\rho)$ and
$\cR_\rho(\cM_s) \subset \cM_s$, though we do not need this property
here and thus we will not prove it.


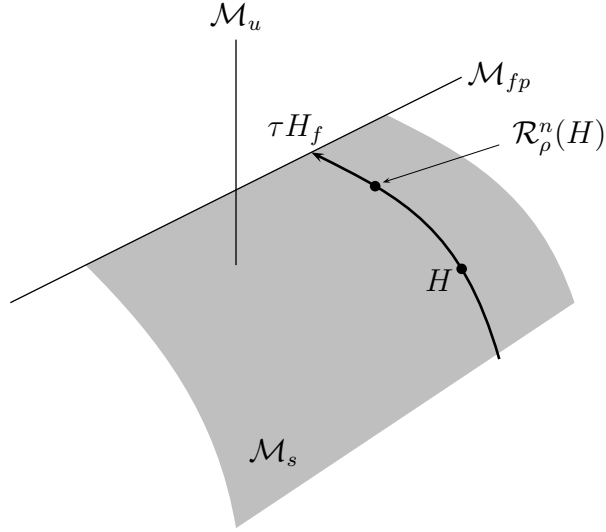
\pspicture(-4,-5)(8,3)

\pscustom[linestyle=none,fillstyle=solid,fillcolor=lightgray]{
\psbezier(2,1)(3,0.5)(4,0)(4.5,-1.5)
\psbezier[liftpen=1](0,-4.5)(-0.25,-3)(-1,-2)(-2,-1)}
\psbezier[linewidth=1pt]{<-}(1,0.5)(2.2,-0.1)(3,-0.5)(3.5,-2.25) 
\rput(0.5,-3.5){$\mathcal{M}_s$}

\psline[linewidth=0.5pt](0,-1)(0,2)\rput(0,2.3){$\mathcal{M}_u$}
\psline[linewidth=0.5pt](-3,-1.5)(3,1.5)
\rput(3.5,1.5){$\mathcal{M}_{fp}$} \rput(0.8,0.8){$\tau H_f$}
\qdisk(1.85,0.05){2pt}

\rput(4.3,0.7){$\mathcal{R}^{n}_{\rho}(H)$}
\psline[linewidth=0.3pt]{<-}(1.95,0.1)(3.5,0.6)
\qdisk(3,-1.05){2pt}\rput(2.7,-1.2){$H$}

\endpspicture


The next result reveals the spectral significance of the map $e$:
\begin{theorem} \label{thm-V.2}
Let $H_s\in \cD_s$. Then the number $e(H_s)$ is
an eigenvalue of the operator $H_s$ and
$\sigma(H_s) \subset e(H_s) +S$ where
\begin{equation}\label{eqn:S}
S:=\{w\in \mathbb{C}| \rRe w \ge 0,
|\rIm w| \le \frac{1}{3} \rRe w \}.
\end{equation}
\end{theorem}
This theorem implies Theorem \ref{thm-main} formulated in the
introduction. We begin with some preliminary results, collected in
Proposition \ref{prop-V.4} below, from which we derive Theorems
\ref{stable-manif} and \ref{thm-V.2}.


\begin{proposition} \label{prop-V.4}
Let $V_{-1}\equiv \cD$   and $e_{-1}(H_s)=0\ \forall H_s$. The
triples $(V_n, E_{n}, e_{n})$, $n=0, 1, ...$, where $V_n$ is a
subset of $\cD$, $E_{n}$ is a map of $V_{n-1}$ into $\mathbb{C}$,
and $e_{n}$ is a map of $\cD_s$ into $\mathbb{C}$, are defined
inductively in $n \geq 0$ by the formulae
\begin{equation}V_n:=\{H\in \cD|\ |H_u+e_{n-1}(H_s)| \leq
\frac{1}{12}  \rho^{n+1}\}, \label{eqn:52aa}
\end{equation}
\begin{equation}
E_{n}(H):= \big(\cR_\rho^{n}(H)\big)_u, \label{eqn:52aaa}
\end{equation}
\begin{equation}e_{n}(H_s)\ \mbox{is the unique zero
of the function}\ E_{n}(H_s-\lambda)\
\label{eqn:52b}
\end{equation}
in the disc $D(e_{n-1}(H_s),\frac{1}{12} \rho^{n+1})$. Moreover,
these objects have the following properties:
\begin{equation}V_{n} \subset V_{n-1}\ \mbox{and}\ V_n \subset
D(\cR_\rho^{n+1}),\label{eqn:52cc}
\end{equation}
$E_{n}(H_s-\lambda)$ is analytic in $\lambda \in
D(e_{n-1}(H_s),\frac{1}{12} \rho^{n+1})$ and in $H_s \in \cD_s$,
$e_{n}(H_s)\in \mathbb{R}$, if $H=H^*$, and
\begin{equation}|\ e_{n}(H_s)-e_{n-1}(H_s)| \leq
2\alpha_n \rho^n.\label{eqn:52dd}
\end{equation}
\end{proposition}

\begin{proof}
We proceed by induction in the index $n$. For $n=0$ the
proposition is trivially true.
We assume that the statements of the proposition hold for all $0
\leq n \leq j-1$ and prove them for $n=j$. Let $e_{n}(H_s)$ and
$E_{n}(H_s-\lambda)$,   $0 \leq n \leq j-1$, be as defined in the
proposition. Since $e_{j-1}(H_s)$ is defined by \eqref{eqn:52b}
with $n=j-1$ we can define $V_j$ using \eqref{eqn:52aa} with
$n=j$. Next, by \eqref{eqn:52cc} with $n=j-1$, $V_{j-1} \subset
D(\cR_\rho^{j})$ and therefore the map $E_{j}$ is well defined.

Let $ H\in V_{j-1}$ and denote $\lambda:=-H_u$ so that $H:=H_s
-\lambda$. Let $H^{(j)}(\lambda):= \cR_\rho^j(H^{(0)}(\lambda))$
with $H^{(0)}(\lambda):=H_s -\lambda$ (\textit{we suppress the
dependence of} $H^{(j)}(\lambda)$  \textit{on}  $H_s $). Write
inductively $H^{(j)}(\lambda):= \cR_\rho(H^{(j-1)}(\lambda))$.

We claim that $H^{(j)}(\lambda) 
$ is analytic  (in the sense specified in the paragraph preceding
Theorem \ref{stable-manif}) in $\lambda \in
D(e_{j-1}(H_s),\frac{1}{12} \rho^{j+1})$ and in $ H_s \in \cD_s $.
We prove this statement by induction in $j$. Clearly,
$H^{(0)}(\lambda)=H_s -\lambda$  is analytic in $\lambda \in
D(e_{-1}(H_s),\frac{1}{12} \rho)$ and in $H_s \in \cD_s$. Now,
assume that $H^{(j-1)}(\lambda)$ is analytic in $\lambda \in
D(e_{j-2}(H_s),\frac{1}{12} \rho^{j})$ and in $H_s \in \cD_s$. Then
by  Proposition~\ref{analytic00}, Appendix III,
$H^{(j-1)}_{0}(\lambda):= E^{(j-1)}(\lambda)+T^{(j-1)}(\lambda)$ and
$W^{(j-1)}(\lambda)$ are analytic. By the properties of
$T^{(j-1)}(\lambda)$, the inverse $
H^{(j-1)}_{0}(\lambda)^{-1}\bchi_\rho$ is well-defined and is
analytic and therefore so is
$$\bchi_\rho H^{(j-1)}(\lambda)_{\tau,\bchi_\rho}^{-1}\bchi_\rho =\sum_{n=0}^\infty \bchi_\rho (-H^{(j-1)}_{0}(\lambda)^{-1}\bchi_\rho W^{(j-1)}(\lambda)\bchi_\rho )^n H^{(j-1)}_{0}(\lambda)^{-1}\bchi_\rho .$$
By the definition of the decimation map, \eqref{IV.1}-\eqref{IV.2},
$$F_\rho (H^{(j-1)}(\lambda)) =H^{(j-1)}_{0}(\lambda)+ \chi_\rho W^{(j-1)}(\lambda)\bchi_\rho  H^{(j-1)}(\lambda)_{\tau,\bchi_\rho}^{-1}\bchi_\rho W^{(j-1)}(\lambda)\chi_\rho ,$$
is analytic. Hence, by the definition of the renormalization
map $\cR_\rho$ in \eqref{IV.5} - \eqref{IV.6}, $\cR_\rho(H^{(j-1)}(\lambda))$ is analytic as well.

This implies that
$E_{j}(H_s-\lambda)$ is analytic in $\lambda \in D(e_{j-1},$
$\frac{1}{12} \rho^{j+1})$ and in $H_s \in \cD_s$.


In the remaining part of the proof we will use the shorthand
$e_{n}\equiv e_{n}(H_s)$ and (abusing notation) $E_{n}(\lambda)\equiv
E_{n}(H_s-\lambda)$. Now, we prove \eqref{eqn:52b} and
\eqref{eqn:52dd} with $n=j$. We begin with some preliminary
estimates. Let $H \in V_{j-1}$. For $1 \leq n \leq j$ denote
\begin{equation}
\Delta_n E(\lambda):=E_{n}(\lambda)-\rho^{-1}E_{n-1}(\lambda).
\label{eqn:53}
\end{equation}
Since $\cR_\rho^{n}(H)= \cR_\rho\big(\cR_\rho^{n-1}(H)\big)$, we
have, by Theorem \ref{thm-III-2-5}, that $|\Delta_n E(\lambda)| \leq
\alpha_n$. This and the analyticity of $\Delta_n E(\lambda)$ in
$D(e_{n-1},\frac{1}{12}  \rho^{n+1})$
together with the Cauchy formula imply that
\begin{equation}
|\partial_{\lambda}^m \Delta_n E(\lambda)|\le\alpha_n (\frac{1}{12}
\rho^{n+1})^{-m}\ for\ n\le j\ \mbox{and}\ m= 0,1.
\label{eqn:54}
\end{equation}
Iterating \eqref{eqn:53} we find for $i \leq j$
\begin{equation}
E_{i}(\lambda)=\rho^{-i}(E_{0 i}(\lambda)-\lambda), \label{eqn:55}
\end{equation}
where
\begin{equation}
E_{0 i}(\lambda):=\sum_{n=1}^i\rho^n\Delta_n E(\lambda).
\label{eqn:56}
\end{equation}
By the estimate  \eqref{eqn:54} with $m=1$
we have for $i \le j$
\begin{equation*}
|\partial_\lambda E_{0 i}(\lambda)|\le
\sum_{n=1}^i\rho^n|\partial_\lambda\Delta_n E(\lambda)|\le
c\sum_{n=1}^i c^{2n-1}\rho^{2\mu(n-1)-2}\gamma_0^2,
\end{equation*}
which, by the conditions on the parameters, \eqref{rho}, implies

\begin{equation}|\partial_\lambda E_{0 i}(\lambda)|\le
c\rho^{-2} \gamma_0^2 \le\frac{1}{5}\label{eqn:57}
\end{equation}
for $0 < i \leq j$.

Now, we are ready to
show the existence and properties of $e_{j}$, stated in
\eqref{eqn:52b} and \eqref{eqn:52dd} with $n=j$, i.e. to show that
$E_{j}(\lambda)$ has a unique zero, $e_{j}$, in every disc
$D(e_{j-1},r \rho^j)$  with $2\alpha_{j} \leq r \leq
\frac{1}{12}\rho$.
The latter is equivalent to showing that $e_{j}$ is a fixed point
of the map $\lambda \rightarrow E_{0 j}(\lambda)$ in the discs
$D(e_{j-1},r \rho^{j})$. Using the equations $ e_{j-1}=E_{0
j-1}(e_{j-1})$
and \eqref{eqn:56} with $i=j-1, j$ and using the triangle inequality
we obtain
$$|E_{0
j}(\lambda)-e_{j-1}| \leq \rho^{j}|\Delta_{j} E(\lambda)|+ |E_{0
j-1}(\lambda)-E_{0 j-1}(e_{j-1})|.$$ Now, remembering the estimate
\eqref{eqn:54} (with $m=0$ and $n=j$) and the estimate
\eqref{eqn:57} (with $i=j-1$) and using the mean-value theorem we
arrive at the inequality
\begin{equation} \label{17}|E_{0
j}(\lambda)-e_{j-1}| \leq
\rho^{j}\alpha_j + \frac{1}{5}|\lambda-e_{j-1}|, \end{equation} and
therefore, $|E_{0 j}(\lambda)-e_{j-1}| \leq r
\rho^{j}$, provided $|\lambda-e_{j-1}| \leq r\rho^{j} $ (remember
that $\alpha_j \le \alpha_0 \ll \rho \ll 1$). This inequality
together with Eqn \eqref{eqn:57} with $i=j$
implies that the map $\lambda \rightarrow E_{0 j}(\lambda)$ has a
unique fixed point, $e_{j}$, in the disc $D(e_{j-1},r \rho^{j})$.
For $r=  \frac{1}{12}\rho$ this gives \eqref{eqn:52b} with $n=j$.
Taking $r=2\alpha_{j}$ we arrive at \eqref{eqn:52dd} with $n=j$.

If $H$ is self-adjoint, then so is the operator $\cR_{\rho}(H)$,
and, consequently, $\cR_{\rho}^j(H)= \cR_{\rho}^j(H)^*$. Hence
$E_{j}(\lambda)$ and $e_{j}$ are real in this case.

Next, we show the first inclusion in \eqref{eqn:52cc} for $n=j$. Let
$H \in V_{j}$ and hence $|\lambda - e_{j-1}|\leq \frac{1}{12}
\rho^{j+1}$. Then, by the induction assumption \eqref{eqn:52dd} for
$n=j-1$, we have that $|\lambda - e_{j-2}|\leq\frac{1}{12}
\rho^{j+1}+ 2\alpha_{j-1} \rho^{j-1}\leq \frac{1}{12} \rho^{j}$ and
therefore
$H \in V_{j-1}$, as claimed.

We proceed to show the second inclusion in \eqref{eqn:52cc} for
$n=j$.
Let $H \in V_{j}$ and keep the notation as above. Since
$E_{j-1}(e_{j-1})=0$, we have that $|E_{j}(\lambda)|\le|\Delta_{j}
E(\lambda)| + \rho^{-1}|E_{j-1}(\lambda)-E_{j-1}(e_{j-1})|$ which by
\eqref{eqn:54}, \eqref{eqn:55} and \eqref{eqn:57} with $i=j-1$ gives
$|E_{j}(\lambda)| \leq \alpha_{j}+
\frac{6}{5}\rho^{-j}|\lambda-e_{j-1}|.$
Hence, since $\alpha_{j} \le \alpha_0$ and by \eqref{rho},
\begin{equation}
|E_{j}(\lambda)|
\le\frac{1}{8}\rho, \label{eqn:58}
\end{equation}
provided $|\lambda - e_{j-1}|\leq \frac{1}{12} \rho^{j+1}$. Thus,
using Theorem \ref{thm-III-2-5} and \eqref{eqn:58} we conclude that,
for $n:=j$,
\begin{equation}
\cR_\rho^{n}(V_{n}) \subset \cD^{\mu, 1}(\rho/8,\beta_n,\gamma_n)
\label{eqn:51a}
\end{equation}
with the numbers $\beta_n$ and $\gamma_n$ given inductively by
$\beta_n=\beta_{n-1}+3C_\chi\frac{\gamma_{n-1}^2}{2\rho}\
\mbox{and}\ \gamma_n=256 C_\chi^2 \rho^\mu\gamma_{n-1}$
and in final form, in the paragraph preceding Theorem VI.1.
Clearly, $\beta_n$, $\gamma_n\le\frac{\rho}{8}$.  E.g.
$\beta_n\le\beta_0+c\frac{\gamma_0^2}{\rho}\lb
1-(c\rho^\mu)^2\rb^{-1}<\frac{\rho}{8}$. Hence, by Lemma
\ref{lem-III-2-2}, $\cR_\rho^{j}(V_{j}) \subset D(\cR_\rho)$.
Thus \eqref{eqn:52cc} is proven for $n=j$.
\end{proof}

\begin{proof}[Proof of Theorem \ref{stable-manif}.]
By \eqref{eqn:52dd}, the
limit $e(H_s):=\lim_{j\rightarrow \infty} e_{j}(H_s)$ exists
pointwise for $H\in \cD$. Iterating Eqn \eqref{eqn:52dd} we find
the estimate
\begin{equation}
|e_{n}(H_s)-e(H_s)| \leq 3\alpha_{n+1} \rho^{n+1}. \label{V.20}
\end{equation}
Given that $\alpha_0 \le \frac{\rho}{108}$ (this is a condition on
the (bare) coupling constant $g$), this inequality implies that
\begin{equation}
V_{n} \subset U_{\delta_n} \subset V_{n-1}. \label{eqn:62}
\end{equation}
where  $\delta_n:=\frac{1}{18}\rho^{n}$.
%
%
%

To prove the analyticity of $e(H_s)$ we note that, since
$E_{j}(\lambda, H_s)$ is analytic in $H_s \in \cD_s$, then so is
$e_{j}(H_s)$. By \eqref{eqn:52dd} the limit
$e(H_s):=\lim_{j\rightarrow \infty} e_{j}(H_s)$ is also analytic in
$H_s \in \cD_s$.

Eqns \eqref{eqn:52cc}
and \eqref{eqn:62} imply the first part of \eqref{eqn:30aaa}. The
second part of \eqref{eqn:30aaa} follows from  Theorem
\ref{thm-III-2-5} and \eqref{eqn:58}.

Now we prove the last statement of Theorem \ref{stable-manif}.
%
%
%
%
%
Let $H\in U_{\delta_n}\subset V_{n} \subset D(\cR_\rho^{n+1})$.
According to \eqref{Hsplit}, $H^{(n)}:=\cR_{\rho}^{n}(H)$ can be
written as
\begin{equation} \label{Hnsplit} H^{(n)}=E_{n}\one + T_{n} +
W_{n},
\end{equation}
where $ T_{n}\equiv T_{n}(H_f)$ with $ T_{n}(r) \in C^1$ and $
T_{n}(0)= 0$. Hence the function $\tau_n(r):=T_{n}(r)/r$  is well
defined. By \eqref{eqn:51a}  we have $|\partial_r T_{n}(r)-1| \leq
\beta_n$ and $\|W_{n}\|_{\cW_{op}^s} \leq \gamma_n$. This gives the
desired estimates for the last two terms in \eqref{eqn:30b}. Let
$E_{n}(\lambda)\equiv E_{n}$ for $\lambda:=-H_u$. To prove the bound
on the first term on the r.h.s. of \eqref{eqn:30b} we
use the relation $E_{n}(e_{n}) =0$ and Eqns \eqref{eqn:55} and
\eqref{eqn:57} to obtain
%
%
\begin{equation} \label{21}
|E_{n}(\lambda)| = |E_{n}(\lambda)-E_{n}(e_{n})|
\leq \frac{6}{5}\rho^{-n}| \lambda - e_{n}|.
\end{equation}
This inequality together with \eqref{V.20} implies,
$|E_{n}(\lambda)|\leq \frac{6}{5} \nu_n +
\frac{18}{5}\alpha_{n+1}\rho \le 2\nu_n, $ provided $|\lambda-e|
\leq \nu_n \rho^n $ and $4 \alpha_n \leq \nu_n $. Finally, if $H$
is self-adjoint, then so is $\cR_{\rho}^{n}(H)$ and therefore
$E_n$ and $\tau_n(r):=T_{n}(r)/r$ are real.

To complete the proof of Theorem \ref{stable-manif} it remains to
show that  as $n \rightarrow \infty$, the functions $\tau_n(r) $
converge in $L^\infty$ to a constant function, $\tau$, as $n
\rightarrow \infty$.
To prove this property requires representing the operators $T^{(n)}$
as sums of the $j$-th step corrections,
\begin{equation}
\Delta_n T(r):=T_{n}(r)-\rho^{-1}T_{n-1}(\rho r), \label{eqn:Tcorr}
\end{equation}similarly to \eqref{eqn:55} and \eqref{eqn:56}. In fact,
this analysis gives that $\tau= \lim_{n \rightarrow \infty}
\tau_n(0)$. We omit the details here but refer the reader to
\cite{BachChenFroehlichSigal2003}.
\end{proof}

\begin{proof}[Proof of Theorem \ref{thm-V.2}.]
It is shown in Appendix II (Section \ref{sec-VII}), Theorem
\ref{thm-XI.1}, that $e(H_s)$ is an eigenvalue of $H_s$ (cf.
\cite{BachFroehlichSigal1998a,BachFroehlichSigal1998b,BachChenFroehlichSigal2003}).
Here we show the second statement of the theorem regarding the
spectrum of $H_s$. As above, we omit the reference to $H_s$ and set
$e\equiv e(H_s)$ and $e_{n}\equiv e_{n}(H_s)$.

We first consider the case of a  self-adjoint operator $H_s$.
%
Let $H^{(n)}(\lambda):=\cR_{\rho}^{n}(H_s-\lambda)$
and, recall, $E_{n}(\lambda):= H^{(n)}(\lambda)_u$.
Eqns \eqref{eqn:55} and \eqref{eqn:57} imply the estimate
$\partial_\lambda E_{n}(\lambda) \le -\frac{4}{5}\rho^{-n}$. Using
the equation $E_{n}(e_{n})=0$, the mean value theorem and the
estimate above, we obtain that $E_{n}(\lambda) \ge
-\frac{4}{5}\rho^{-n}(\lambda -e_{n})$, provided $\lambda\le e_{n}$.
Hence, if $\lambda\le e_{n} - \theta_n$, with $\theta_n \gg \gamma_n
\rho^n$ and $\theta_n \rightarrow 0$ as $n \rightarrow \infty$, then
$H^{(n)}(\lambda)\ge \frac{4}{5}\rho^{-n}\theta_n - O(\gamma_n) \ge
\frac{1}{2}\gamma_n$. This implies $0 \in \rho(H^{(n)}(\lambda))$
and therefore, by Theorem \ref{thm-II-1}, $0 \in \rho(H_s -
\lambda)$ or $\lambda \in \rho(H_s)$.
%
%
Since $e_{n} \rightarrow e$ and $ \theta_n \rightarrow 0$ as $n
\rightarrow \infty$, this implies that $\sigma(H_s) \subset [e,
\infty)$, which is the second statement of the theorem for
self-adjoint operators.

Now we consider a non-self-adjoint operator $H_s$. For all $n \ge
0$, we have shown that if $H_s \in \cD_s,\ e=e(H_s)$ and if $\mid
\lambda - e \mid \le \delta_n$, where $\delta_n = \nu_n \rho^n$,
then $H_s - \lambda \in \dom(\cR_{\rho}^{n})$ and $ H^{(n)}(\lambda)
: = \cR_{\rho}^{n} (H_s - \lambda) \in D^{\mu, 1} (\rho  \slash8,
\beta_n, \gamma_n)$. By Theorem \ref{thm-II-1} we have that
\begin{equation} \label{V.26}
\lambda \in \sigma (H_s)\quad \Leftrightarrow\quad 0 \in \sigma (H^{(n)}
(\lambda)),
\end{equation}
if $\mid \lambda - e \mid \le \delta_n$. By Theorem
\ref{stable-manif}, we can decompose
\begin{equation} \label{V.27}
H^{(n)}(\lambda) = E_n (\lambda) + \tau_n (H_f, \lambda) H_f + W_n
(\lambda),
\end{equation}
with $\|W_n (\lambda) \| \le \gamma_n$  on $\hbox{Ran} \chi_{H_f \le
\rho}$.  Hence
\begin{equation} \label{V.28}
0 \in \sigma \big(H^{(n)} (\lambda)\big) \Rightarrow\ \exists r \in
[0, \rho]:\ |E_n (\lambda) + \tau_n (r, \lambda) r| \le \gamma_n.
\end{equation}
Using that $E_n(e_n) =0$ ($e_n\equiv e_n(H_s))$ and 
the integral of derivative formula we find
\begin{equation} \label{V.29}
E_n (\lambda) = (\lambda - e_n)g(\lambda)
\end{equation}
with $g(\lambda):= \int_0^1 E'_n (e_n + s(\lambda -e_n))ds$. Note
that $ \bar{\lambda}: =e_n + s(\lambda -e_n)$ satisfies $\mid
\bar{\lambda} - e \mid \le \delta_n$ for $0 \le s \le 1$. This and
\eqref{eqn:55}, \eqref{eqn:57} and $\rho^{-1} \gamma_0 \ll 1$ imply
that
\begin{equation} \label{V.31}
\mid  g (\lambda) + \rho^{-n} \mid\ \le \frac{1}{5}\rho^{-n}.
\end{equation}
In addition, below we use the estimate \eqref{V.20} which we rewrite
as:
\begin{equation} \label{V.31'}
 \mid e_n - e \mid \le 3 \alpha_{n+1} \rho^{n+1}.
\end{equation}

We denote $\mu : = \lambda - e$ so that
\begin{equation} \label{V.30}
E_n (\lambda) = g(\lambda) (\mu + e - e_n).
\end{equation}
We consider separately two cases.
%

a) $\rRe \mu \le - \theta\ \hbox{and} \mid \hbox{Im} \mu \mid \le 3
\theta$ with $\theta \ge 36 \alpha_{n+1} \rho^{n+1}$. Using
\begin{equation*}
\rRe (E_n +  \tau_n r) = \rRe g \rRe \mu - \rIm g \rIm \mu + \rRe (g
(e - e_n)) + \rRe \tau_n  r
\end{equation*}
and using \eqref{V.31}, we obtain
\begin{equation*}
\rRe (E_n + \tau_n r) \ge \frac{4}{5} \rho^{-n} \theta - \frac{3}{5}
\rho^{-n} \theta
\end{equation*}
$$- \frac{6}{5} \rho^{-n} 3 \alpha_{n+1} \rho^{n+1}
+(1 - \beta_n) r.$$ Since $\theta \ge 36 \alpha_{n+1} \rho^{n+1}$
this gives
\begin{equation} \label{V.32}
\rRe(E_n+\tau_nr)\ge \frac{1}{10} \rho^{-n} \theta+(1-\beta_n)r.
\end{equation}

b) $|\hbox{Im} \mu| \ge \theta$ and $|\hbox{Re} \mu| \le 3 \theta$.
If $r\le 10 \theta \rho^{-n}$, then
\begin{eqnarray*}
| \rIm ( E_n+\tau_n r)| &=& \rRe  g \rIm \mu +\rIm g \rRe \mu+
\rIm(E'_n(e-e_n))\\
&+& \rIm \tau_n r| \ge \frac{4}{5} \rho^{-n}\theta -\frac{3}{5}
\rho^{-n}\theta\\
&-& \frac{6}{5}\rho^{-n} 3\ \alpha_{n+1} \rho^{n+1} - \beta_n
10\theta \rho^{-n}.
\end{eqnarray*}
This gives
\begin{equation} \label{V.33}
| \rIm (E_n + \tau_n r)| \ge \frac{1}{10} \theta \rho^{-n},
\end{equation}
provided $\theta \ge 72 \alpha_{n+1}\rho^{n+1}$ and $\beta_n\le
10^{-3}.$

Now, if $r \ge 10\theta \rho^{-n}$, then we estimate by \eqref{V.30}
and \eqref{V.31}
\begin{equation} \label{V.34}
| E_n+\tau_n r| \ge | g \mu +\tau_n r| -\frac{6}{5} \rho^{-n}
3\alpha_{n+1} \rho^{n+1}.
\end{equation}
Furthermore, we have
\begin{eqnarray*}
|g\mu +\tau_n r|^2 &=&(\rRe g \rRe \mu -\rIm g \rIm \mu +
\rRe \tau_n r)^2\\
&+&(\rRe g \rIm \mu +\rIm g \rRe \mu + \rIm \tau_n r)^2\\
&\ge& (\rRe g \rRe \mu +\rRe \tau_n r)^2 - (\frac{1}{5} \rho^{-n}
\rIm \mu)^2\\ &+& \frac{4}{5} \rho^{-n} | \rIm \mu| -  \frac{3}{5}
\rho^{-n} | \rIm \mu|^2 - (\beta_n r)^2.
\end{eqnarray*}
Since $\rRe g \rRe \mu +  \frac{1}{2} \rRe \tau_n r \ge 0$, we have
$| g \mu + \tau_n r |^2 \ge (\frac{1}{2} \rRe \tau_n r)^2 - (\beta_n
r)^2$
which gives, for $\theta \le 10^{-3}$,
\begin{align*}
|g \mu + \tau_n r| &\ge \frac{1}{2} (1 - 2 \beta_n)
r \\
&\ge 5 (1-2 \beta_n) \theta \rho^{-n} \ge 2 \theta \rho^{-n}.
\end{align*}
This together with \eqref{V.34} yields
$\mid E_n + \tau_n r \mid \ge \theta \rho^{-n},$
provided $\theta \ge 4 \alpha_{n+1} \rho^{n+1}$.  This together with
\eqref{V.33} gives for the case b)
\begin{equation} \label{V.36}
| E_n + \tau_n r |\ \ge \frac{1}{10}  \theta \rho^{-n},
\end{equation}
provided $\theta \ge 72 \alpha_{n+1} \rho^{n+1}\ \hbox{and}\ \theta
\le 10^{-3}$.

The inequalities \eqref{V.32} and \eqref{V.36} and relations
\eqref{V.26} and \eqref{V.28} show that $\lambda \in \rho (H_s)$ if
either $\rRe \mu \le - \theta$ and $|\rIm \mu | \le 3 \theta$
or $\rIm \mu \ge \theta$ and $| \rRe\mu | \le 3 \theta$
with $\mu : = \lambda - e$, provided
\begin{equation} \label{V.37}
\theta \ge \max(20 \rho^n \gamma_n, 72 \alpha_{n+1} \rho^{n+1})\
\hbox{and}\ \beta_n \le 10^{-3}.
\end{equation}
This can be written as
\begin{equation} \label{V.36'}
\Omega^{(1)}_{\theta}, \Omega^{(2)}_{\theta} \subset \rho (H_s)
\end{equation}
where $\theta$ satisfies \eqref{V.37} and
$$
\Omega^{(1)}_{\theta} : = \{\lambda \in \mathbb{C} \mid \rRe \mu \le - \theta\ \hbox{and}\
\mid\rIm \mu\mid \le 3 \theta\}
$$
and
$$\Omega^{(2)}_{\theta} : = \{ \lambda
\in \mathbb{C} \mid \mid \rIm \mu \mid \ge \theta\ \hbox{and}\ \mid
\rRe \mu \mid \le 3 \theta\}.$$ 

Define the new subset
$\Omega^{(3)}_{\theta} : = \{\lambda \in \mathbb{C} \mid \rRe \mu
\le - \theta\}.$
We claim that \begin{equation} \label{V.37'}\Omega^{(3)}_{\theta}
\subset \cup _{n=0}^\infty (\Omega^{(1)}_{3^{n}\theta} \cup
\Omega^{(2)}_{3^{n+1} \theta} ).\end{equation}
 Indeed,
$\Omega^{(3)}_{\theta} /(\Omega^{(1)}_{\theta} \cup \Omega^{(2)}_{3
\theta} )\bigcap\Omega^{(3)}_{\theta} \subset \{\lambda \in
\mathbb{C} \mid \rRe \mu \le - 3\theta,\ |\rIm \mu | \ge 3\theta \}
\subset \Omega^{(3)}_{3\theta}$ and therefore $\Omega^{(3)}_{\theta}
\subset \Omega^{(1)}_{\theta} \cup \Omega^{(2)}_{3 \theta}\cup
\Omega^{(3)}_{3\theta}.$ Iterating the last inclusion we arrive at
the desired relation. Eqns \eqref{V.36'} and \eqref{V.37'} imply
that
\begin{equation*}
\Omega^{(2)}_{\theta} \cup \Omega^{(3)}_{\theta} \subset \rho (H_s).
\end{equation*}
for any $\theta$ satisfying \eqref{V.37}.

Now assume $\lambda \notin e + S$, where $S$ is defined in
\eqref{eqn:S}. Then either $\rRe \mu < 0$ or $\rRe \mu \ge 0$ and
$|\rIm \mu|
> \frac{1}{3}  \rRe \mu$.  In the first case $\exists n$ s.t. $\rRe
\mu < - \theta_n$ where $\theta_n : = \max (20 \rho^n \delta_n, 72
\delta_{n+1} \rho^{n+1}),$ and therefore $\lambda \in
\Omega^{(3)}_{\theta_n} \subset \rho (H_s)$.  In the second case,
assuming $\mu > 0$, we choose $n$ s.t. $\rRe \mu \approx 3
\theta_n$. Then $| \rIm \mu |\ \ge \theta_n$ and $| \rRe \mu |\
\le 3 \theta_n$ so that $\lambda \in \Omega^{(2)}_{\theta_n}
\subset \rho (H_s)$.  Hence $\mathbb{C} \slash \{ e + S\} \subset
\rho (H_s)$ which implies $\sigma(H_s) \subset e(H_s) +S$.
\end{proof}

\begin{remark} \label{rem-2} Define $E_{0\infty}(e(H_s), H_s):= \lim_{j\rightarrow \infty} E_{0j}(e(H_s),
H_s)$. Then
\begin{equation}
E_{0\infty}(e(H_s), H_s)= \sum_{i=1}^\infty\rho^i\Delta_i E(e(H_s),
H_s), \label{eqn:60}
\end{equation}
where the series on the right hand side converges absolutely by
estimate \eqref{eqn:54}, and $e(H_s)$ satisfies the relation
\begin{equation}
e(H_s)=E_{0\infty}(e(H_s), H_s). \label{eqn:61}
\end{equation}
Indeed, Eqns \eqref{eqn:55} and \eqref{eqn:58} yield $|E_{0
j}(\lambda, H_s)-\lambda|
\leq \frac{1}{8} \rho^j,$ provided $|\lambda - e_{j-1}(H_s)| \leq
\frac{1}{12} \rho^{j+1}$, which together with \eqref{V.20}
implies \eqref{eqn:61}.

Eqns \eqref{eqn:60} - \eqref{eqn:61}, \eqref{eqn:55} and
\eqref{eqn:54} (with $m=0$) imply that
\begin{align} \nonumber
|E_{0 n}(\lambda)-e|&\le|E_{0 n}(\lambda)-E_{0 n}(e)|+|E_{0
n}(e)-E_{0\infty}(e)|\\
&\le \sup_{\lambda\in A_{\delta_n}
}(|E_{0 n}'(\lambda)|)|\lambda-e|+\sum_{i=n+1}^\infty
\rho^{i}\alpha_i.
\end{align}
Now, using Eqn \eqref{eqn:57} and the definition of $\alpha_i$ we
obtain, furthermore, that
\begin{equation}
|E_{0 n}(\lambda)-e| \le\frac{1}{5}|\lambda-e|+
(1-\rho)^{-1}\rho^{n+1} \alpha_{n+1}.
\end{equation}
This estimate is used in our further work,
\cite{FroehlichGriesemerSigal2008b}.
\end{remark}

%

\secct{Appendix I:
Proof of Theorem \ref{thm-III-2-5}} \label{sec-VIII}
%


%
The proof below is similar to and relies on some parts of the
corresponding proof in \cite{BachChenFroehlichSigal2003}. We proceed
in two steps. First we determine $\huw$ s.t. $H(\huw) =:
\cR_\rho(H(\uw))$. In fact, we find explicit formulae expressing
$\huw$ in terms of $\uw$. Then, using these formulae, we estimate
$\huw$.
%

Let $H(\uw)  \  \in \cD^{\mu,0}(\rho/8, 1/8, \rho/8)$. We write this
operator as $H(\uw)  \ = \   H_0 \; + \; W$ where $H_0:=E+T$.
According to the definition (Eqns \eqref{II-3}) and \eqref{IV.2}) of
the smooth Feshbach map, $F_{\rho}$, we have that
\begin{eqnarray} \label{eq-III-2-18}
\lefteqn{ F_{\rho} \big( H(\uw) \big) \ = \   H_0 \; + \; \chi_\rho
W \chi_\rho }
\\[1mm] \nonumber & & \hspace{-3mm}
\; - \; \chi_\rho \, W \, \bchi_\rho \big(H_0 + \, \bchi_\rho W
\bchi_\rho \big)^{-1} \bchi_\rho \, W \, \chi_\rho \period
\end{eqnarray}
%
%
%
Here, recall, the cut-off operators $\chi_\rho \equiv \chi_{H_f \le
\rho}$ are defined in Section \ref{sec-III} and $\bchi_\rho:= 1 -
\chi_\rho$. Note that, because of $H(\uw) \in \cD^{\mu,0}(\rho/8,
1/8, \rho/8)$
and of \eqref{eq-III-1-25.1}
%
\begin{equation} \label{eq-III-2-19}
\big\|  H_0^{-1} \bchi_\rho^2 \big\| \ \leq \ \frac{2}{\rho}
\hspace{5mm} \mbox{and} \hspace{5mm} \| W \| \ \leq \ \frac{\xi \,
\rho}{8} \period
\end{equation}
Eq.~(\ref{eq-III-2-19}) implies that the Neumann series expansion in
$W_{\bchi_\rho}:= \bchi_\rho W \bchi_\rho$ of the resolvent   in
(\ref{eq-III-2-18}) is norm convergent and yields
\begin{equation} \label{eq-III-2-20}
 F_{\rho} \big( H(\uw) \big) \ = \ H_0  + \sum_{L=1}^\infty
(-1)^{L-1} \, \chi_\rho W \bigg(  H_0^{-1}\bchi_\rho^2 \; W
\bigg)^{L-1} \chi_\rho \period
\end{equation}

To write the Neumann series on the right side of (\ref{eq-III-2-20})
in the generalized normal form we use Wick's theorem,
which we formulate now. We begin with some notation. We introduce
the operator families
\begin{eqnarray} \label{VIII.4}
\lefteqn{ W_{p,q}^{m,n} \big[ \uw \big| r ; \, k_{(m,n)} \big] \ :=
\chi_1 \;\int_{B_1^{p+q}} \frac{ dx_{(p,q)} }{ |x_{(p,q)}|^{1/2} }
\; a^*( x_{(p)} ) \, \times }
\\ \nonumber & & \hspace{-7mm}
 w_{m+p,n+q} \big[ \hf + r ; \, k_{(m)}, x_{(p)} \, , \,
                       \tk_{(n)}, \tx_{(q)} \big] \, a( \tx_{(q)} )
\: \chi_1 \comma
\end{eqnarray}
for $m+n \geq 0$ and a.e $k_{(m,n)} \in B_1^{m+n}$. Here we use the
notation for $x_{(p,q)}$, $x_{(p)}$, $\tx_{(q)}$, etc. similar to
the one introduced in Eqs.~(\ref{III.2})--(\ref{III.4}). For $m=0$
and/or $n=0$, the variables $k_{(0)}$ and/or $\tk_{(0)}$ are dropped
out.
%
Denote by $S_m$ the group of permutations of $m$ elements. Define
the symmetrization operation as
\begin{eqnarray} \label{eq-III-2-24}
\lefteqn{ w_{m,n}^{(\sym)}[ r ; \, k_{(m,n)} ] }
\\ \nonumber & := &
\frac{1}{m! \, n!} \sum_{\pi \in S_m} \sum_{\tpi \in S_n} w_{m,n} [
r ; \, k_{\pi(1)}, \ldots, k_{\pi(m)} \, ; \, \tk_{\tpi(1)}, \ldots,
\tk_{\tpi(n)} \, ] .
\end{eqnarray}
Finally, below we will use the notation
%
%
%
\begin{eqnarray}  \label{eq-III-1-6.2} && \Sigma[ k_{(m)} ]  :=  |k_1| + \ldots +
|k_m| ,\\ \label{eq-III-2-25.2} &&  k_{(M,N)} =
(k^{(1)}_{(m_1,n_1)}, \ldots, k^{(L)}_{(m_L,n_L)}) \comma
\hspace{5mm} k^{(\ell)}_{(m_\ell,n_\ell)} \ \; = \ \;
(k^{(\ell)}_{(m_\ell)}, \tk^{(\ell)}_{(n_\ell)}) \comma
\\ \label{eq-III-2-26}
&& r_\ell :=  \Sigma[\tk^{(1)}_{(n_1)}]  + \ldots +
\Sigma[\tk^{(\ell-1)}_{(n_{\ell-1})}] \, + \,
\Sigma[k^{(\ell+1)}_{(m_{\ell+1})}] + \ldots +
\Sigma[k^{(L)}_{(m_L)}] \comma \hspace{11mm}
\\  \label{eq-III-2-27}
&& \tr_\ell  :=  \Sigma[\tk^{(1)}_{(n_1)}]  + \ldots +
\Sigma[\tk^{(\ell)}_{(n_{\ell})}] \, + \,
\Sigma[k^{(\ell+1)}_{(m_{\ell+1})}]  + \ldots +
\Sigma[k^{(L)}_{(m_L)}],
\end{eqnarray}
with $r_\ell =0$ if $n_1 =\ldots n_{\ell-1} = m_{\ell+1} = \ldots
m_L =0$ and similarly for $\tr_\ell$ and $m_1 + \ldots + m_L = M,\
n_1 + \ldots + n_L = N$.
%
%
%
\begin{theorem}[Wick Ordering] \label{thm-III-2-3}
Let $\uw = (w_{m,n})_{m+n \geq 1} \in \cW_{1}^s$ and $F_j\equiv
F_j(H_f), j= 0 \ldots L,$ where the functions $F_j(r)$ are $C^s$ and
are bounded together with their derivatives.
Write $W := \sum_{m+n \geq 1} W_{m,n}$ with $W_{m,n} :=
W_{m,n}[w_{m,n}]$. Then
\begin{equation} \label{eq-III-2-23}
F_0 \, W \, F_1 \, W \cdots W \, F_{L-1} \, W \, F_L \ = \ \tW
\comma
\end{equation}
where $\tW:=\tW[ \tuw ],\ \tuw := (\tw_{M,N}^{(\sym)})_{M+N \geq 0}
$ with $\tw_{M,N}^{(\sym)}$ given by the symmetrization w.~r.~t.\
$k_{(M)}$ and $\tk_{(N)}$, of the coupling functions
\begin{eqnarray}  \nonumber
\lefteqn{ \tw_{M,N}[ r ; \, k_{(M,N)} ]  =  \sum_{m_1 + \ldots + m_L
= M, \atop n_1 + \ldots + n_L = N} \sum_{p_1, q_1, \ldots, p_L, q_L:
\atop m_\ell + p_\ell + n_\ell + q_\ell \geq 1} \prod_{\ell = 1}^L
\bigg\{ {m_\ell + p_\ell \choose p_\ell} {n_\ell + q_\ell \choose
q_\ell} \bigg\}}
\\ \nonumber
\\ \nonumber
& \hspace{-6mm} & F_0[r+\tr_0] \, \bigg\la \Om \bigg| \, \tW_1 \big[
r+r_1 ; \, k^{(1)}_{(m_1, n_1)} \big] \; F_1[\hf+r+\tr_1] \; \tW_2
\big[ r+r_2 ; \, k^{(2)}_{(m_2, n_2)} \big]
\\  \label{eq-III-2-25}
& \hspace{-6mm} &  \cdots F_{L-1}[\hf+r+\tr_{L-1}] \; \tW_L \big[
r+r_L ; \, k^{(L)}_{(m_L, n_L)} \big] \; \Om \bigg\ra
F_L[r+\tr_L]\comma
\end{eqnarray}
with
%
\begin{eqnarray} \label{VIII.12}
\tW_\ell \big[ r ; \, k_{(m_\ell, n_\ell)} \big] & := &
W_{p_\ell,q_\ell}^{m_\ell,n_\ell} [ \uw \big| \: r ; \, k_{(m_\ell,
n_\ell)} ].
\end{eqnarray}
\end{theorem}
For a proof of this theorem see
\cite[Theorem~A.4]{BachFroehlichSigal1998b}. Here we sketch the idea
of this proof. Substituting the expansion $W := \sum_{m+n \geq 1}
W_{m,n}$ into \eqref{eq-III-2-23} we find
$$\tW= \sum_{m'_1, n'_1, \ldots, m'_L, n'_L \atop m'_\ell
+ n'_\ell  \geq 1}F_0 \prod_{i = 1}^L \bigg (W_{m'_i,n'_i} F_i
\bigg).$$ Now we want to transform each product on the r.h.s. to the
generalized normal form, see Eqn \eqref{III.12}. Each factor has the
creation and annihilation operators entering it explicitly and
through the operators $H_f$. We do not touch the latter and
reshuffle the former.

We pull the annihilation operators, $a$, entering the
$W_{m'_i,n'_i}$'s explicitly, to the left and the creation
operators, $a^*$, to the left. The creation and annihilation
operators interchange positions according to the formula
$$a(k)a^*(k') = a^*(k')a(k) +\delta(k-k').$$
Thus they either pass through each other without a change or produce
the $\delta-$function (contract with each other). Furthermore, they
pass through functions of the photon Hamiltonian operator $H_f$
according to the Pull-Through formulae
%
\begin{equation} \label{eq-III-1-10}
a(k) \, F[\hf] \ = \ F[ \hf + |k| ] \, a(k),\
F[\hf] \, a^*(k) \ = \ a^*(k) \, F[ \hf +
|k| ]  \comma
\end{equation}
which hold on $\cH_\red$ in the sense of operator-valued
distributions for every bounded and measurable function $F$, see
\cite[Lemma~A.1]{BachFroehlichSigal1998b}.

Some of the creation and annihilation operators reach the extreme
left and right positions, while the remaining ones contract. The
terms with $M$ creation operators reaching the extreme left
positions and $N$ annihilation operators reaching the extreme right
positions contribute to the $(M,N)-$ formfactor, $\tw_{M,N}$, of the
operator $\tW$.

This is the standard way for proving the Wick theorem on the
reduction of operators on Fock spaces to their normal (or Wick)
forms, modified by presence of $H_f-$ dependent factors. The problem
here is that the number of terms generated by various contractions,
which is the number of pairs which can be formed by creation and
annihilation operators, is, very roughly, of order of $L!$ for a
product of $L$ terms. Therefore a simple majoration of the series
for $\tw_{M,N}$ will diverge badly. Thus we have to re-sum this
series in order to take advantage of possible cancelations. The
latter is done by, roughly, representing, for a given $M$ and $N$,
the sum over all contractions by a vacuum expectation which effects
only the 'contracting' creation and annihilation operators and does
not apply to the 'external' ones, i.e. those which reached the
extreme positions on the left and right.


As a direct consequence of Theorem~\ref{thm-III-2-3} and
Eqns.~(\ref{IV.6}), (\ref{eq-III-2-5})--(\ref{eq-III-2-6}) and
(\ref{eq-III-2-20}), we find a sequence $\huw$ such that $H(\huw) =
\cR_\rho(H(\uw))= S_\rho \big(F_{\rho} ( \, H(\uw)\, ) \big)$ as
follows.

\begin{theorem} \label{thm-III-2-4}
Let
$H(\uw) \in \cD^{\mu,s}(\rho/8,
\rho/8, \rho/8)$ . Then $\cR_\rho(H(\uw))$ $= H(\huw)$ where $\huw =
(\hw_{M,N}^{(sym)})_{M+N \geq 0}$ with $\hw_{M,N}^{(sym)}$, the
symmetrization w.~r.~t.\ $k^{(M)}$ and $\tk^{(N)}$ (as in
Eq.~(\ref{eq-III-2-24})) of the kernels
\begin{eqnarray} \label{VIII.14}
\lefteqn{ \hw_{M,N}[ \, r ; \, k_{(M,N)} ] \ = \ \rho^{M+N-1}
\,\sum_{L=1}^\infty (-1)^{L-1} \: \times }
\\ \nonumber
& \hspace{-5mm} & \sum_{m_1 + \ldots + m_L = M, \atop n_1 + \ldots +
n_L = N} \sum_{p_1, q_1, \ldots, p_L, q_L:
      \atop m_\ell + p_\ell + n_\ell + q_\ell \geq \delta_L}
\prod_{\ell = 1}^L \bigg\{ {m_\ell + p_\ell \choose p_\ell} {n_\ell
+ q_\ell \choose q_\ell} \bigg\} \;  V_{\umpnq} [ r ; k_{(M,N)} ],
\end{eqnarray}
for $M+N \geq 1$, and
\begin{eqnarray} \label{VI.15}
\hw_{0,0}[ \, r ] \ = r +\ \rho^{-1} \,\sum_{L=2}^\infty (-1)^{L-1}
%
\sum_{p_1, q_1, \ldots, p_L, q_L:
      \atop p_\ell +  q_\ell \geq 1}
\prod_{\ell = 1}^L
 \;  V_{\upq} [ r ],
\end{eqnarray}
%
%
for $M=N =0$. Here
%
$\umpnq := (m_1, p_1, n_1, q_1,$ $\ldots,$ $m_L, p_L, n_L, q_L) \in
\NN_0^{4L}$, and
\begin{eqnarray} \label{VIII.15}
\lefteqn{ V_{\umpnq} [ r ; k_{(M,N)} ] \ := \ }
\\ \nonumber
& & \hspace{-6mm} \bigg\la \Om , \; F_0[\hf+r] \, \prod_{\ell = 1}^L
\Big\{ \tW_\ell \big[ \rho(r+r_\ell) ;
               \, \rho k^{(\ell)}_{(m_\ell, n_\ell)} \big]
\; F_\ell[\hf+r] \Big\} \; \Om \bigg\ra ,
\end{eqnarray}
with $M := m_1 + \ldots + m_L$,  $N := n_1 + \ldots + n_L$, $F_0[r]
:= \chi_1[r+\tr_0]$, $F_L[r] := \chi_1[r+\tr_L]$ and
%
$F_\ell[r] \ := \ \frac{ \bchi_1[r+\tr_\ell]^2 }{ T[\rho
(r+\tr_\ell)] + E } \comma$
%
for $\ell = 1, \ldots, L-1$.  Here
the notation introduced in Eqs.~\eqref{VIII.4}--\eqref{eq-III-2-27}
and \eqref{VIII.12} is used.
\end{theorem}
%

We remark that Theorem~\ref{thm-III-2-4} determines $\huw$ from $\uw
 \in \cW^{\mu,s}$
only as a sequence of integral kernels that define an operator in
$\cB[\cF]$. Now we show that $\huw \in \cW^{\mu,s}$, i.e. $\| \huw
\|_{\mu, s,\xi} < \infty$. In what follows we use the notation
introduced in Eqs.~\eqref{VIII.4}--\eqref{eq-III-2-27}
and \eqref{VIII.12}. To estimate $\huw$, we start with the following
preparatory lemma
\begin{lemma} \label{lem-VIII.3}
For fixed $L \in \NN$ and $\umpnq \in \NN_0^{4L}$, we have
$V_{\umpnq} \in \cW_{M,N}^{\mu,s}$ and
\begin{eqnarray} \label{VIII.16}
\big\| V_{\umpnq} \|_{\mu,s}
\leq 4 C_{\chi}^{2} \, \rho^{\mu}L^s \,
\Big(\frac{C_{\chi}}{\rho}\Big)^{L-1}\; \prod_{\ell = 1}^L \frac{
\big\| w_{m_\ell+p_\ell, n_\ell+q_\ell}
       \big\|_{\mu,s} }
     { \sqrt{p_\ell^{p_\ell} \, q_\ell^{q_\ell}} } \comma
\end{eqnarray}
with the convention that $p^p := 1$ for $p=0$.  Here the constant
$C_{\chi}$ is given by \eqref{Cchi}.
%
%
\end{lemma}
This lemma is proven in \cite{BachChenFroehlichSigal2003} (Lemma
III.10) for
the $L^2-$version of the norms \eqref{III.5} and \eqref{III.7} with
$s=1$
The extension of this lemma to the norms \eqref{III.5} and
\eqref{III.7} with $s=2$, used in this paper, is straightforward. We
present here the proof for $s=0$ and point out how it extends to the
$s>0$ case in order to illustrate its simple structure and for
references needed later.

\begin{remark} \label{rem-VI.4}
The proof of Lemma~\ref{lem-VIII.3} requires taking derivatives of
$\chi_1[r]$ and $\bchi_1[r]$. Here the main advantage of using the
smooth Fesh\-bach map, rather than the (projection) Feshbach map,
becomes manifest. If $\chi_1[r]$ and $\bchi_1[r]$ were projections,
i.e., characteristic functions of intervals, we would inevitably
encounter $\delta$-distributions. In fact, the appearance of these
$\delta$-distributions are the reason for using (a rather involved
mixture of) supremum and $L^1$-norms in
\cite{BachFroehlichSigal1998a,BachFroehlichSigal1998b}. In contrast,
the proof of Lemma~\ref{lem-VIII.3}  is quite straightforward and
merely requires summation of geometric series.
\end{remark}
\Proof First we note that by the definition of the cut-of function
$\chi_1(r) \equiv \chi_{r \le 1}$ (see the paragraph after
\eqref{III.9}), $|F_i[r]| \leq 1,\ i=0, L$.
Moreover, since $T(r) \ge \frac{7}{8}r,\ \supp\bchi_1 \subset \{r
\ge 1\}$ and $|E| \le \frac{1}{8}\rho$, we have that, for $\ell = 1,
\ldots, L-1$,
\begin{eqnarray} \label{eq-III-3-5}
\big|F_\ell[r]\big|
 \leq  \bigg|\frac{ \bchi_1^2[r+\tr_\ell] }{T[\rho (r+\tr_\ell)] -
E} \bigg|
\leq  \frac{4}{3\rho} \period
\end{eqnarray}
%
%
%
%

Now, we estimate $|V_{\umpnq}|$, using that $|\la \Om , A \Om \ra |
\leq \|A\|_\op$, for any $A \in \cB[\cH_\red]$. We have that
\begin{eqnarray} \label{eq-III-3-6}
\lefteqn{ \big| V_{\umpnq} [ r ; k_{(M,N)} ] \big| }
\\ \nonumber
& \leq & \prod_{\ell = 0}^L \big\| F_\ell[\hf+r] \big\|_\op \;
\prod_{\ell = 1}^L \Big\| \tW_\ell \big[  \rho(r+r_\ell) ; \, \rho
k^{(\ell)}_{(m_\ell, n_\ell)} \big] \Big\|_\op \period
\end{eqnarray}
Using \eqref{III.17} and letting $\ell_j$ to be defined by the
property that the vector $k^{(\ell_j)}_{(m_{\ell_j}, n_{\ell_j})}$
contains $k_j$ among its $3-$dimensional components,
we arrive at
%
$$ \big\| V_{\umpnq} \big\|_{\mu}  =
\max_j \sup_{r \in I, k_{(M,N)} \in B_1^{M+N}} \big| |
k_j|^{-\mu}V_{\umpnq} [ r ; k_{(M,N)} ] \big|$$
$$\leq
(\frac{4}{3\rho})^{L - 1}\; \max_j\prod_{\ell \neq \ell_j}^{1,L}
\bigg\{
 \sup_{r \in I, k^{(\ell)}_{(m_\ell, n_\ell)} \in B_1^{m_\ell +
n_\ell}}  \Big\| \tW_\ell \big[
 \,\rho  r ; \, \rho k^{(\ell)}_{(m_\ell, n_\ell)} \big] \Big\|_\op
\; \bigg\}$$
$$\sup_{r \in I, k^{(\ell_j)}_{(m_{\ell_j}, n_{\ell_j})} \in B_1^{m_{\ell_j} +
n_{\ell_j}}}  | k_j|^{-\mu} \Big\|\tW_{\ell_j} \big[ \,\rho  r ;
\,\rho k^{(\ell_j)}_{(m_{\ell_j}, n_{\ell_j})} \big] \Big\|_\op$$
$$\le(\frac{4}{3\rho})^{L - 1} \rho^{\mu} \max_j\prod_{\ell \neq \ell_j}^{1,L}
\bigg\{
 \sup_{r \in I, k^{(\ell)}_{(m_\ell, n_\ell)} \in B_1^{m_\ell +
n_\ell}}  \Big\| \tW_\ell \big[
 r ; \, k^{(\ell)}_{(m_\ell, n_\ell)} \big] \Big\|_\op
\; \bigg\}$$
\begin{eqnarray} \label{eq-III-3-14}
\sup_{r \in I, k^{(\ell_j)}_{(m_{\ell_j}, n_{\ell_j})} \in
B_1^{m_{\ell_j} + n_{\ell_j}}}  | k_j|^{-\mu} \Big\|\tW_{\ell_j}
\big[ r ; \, k^{(\ell_j)}_{(m_{\ell_j}, n_{\ell_j})} \big]
\Big\|_\op.
\end{eqnarray}
We now convert the operator norms on the right side of
(\ref{eq-III-3-14}) into the coupling functions norms. To this end
we use, pointwise in $k^{(\ell)}_{(m_\ell, n_\ell)}$ a.e.,
inequality~(\ref{III.11}) in Theorem~\ref{thm-III.1}  to obtain for
any $\mu \ge 0$
%
%
%
\begin{eqnarray*}
\lefteqn{\max_j \sup_{r \in I,\ k^{(\ell)}_{(m_\ell, n_\ell)} \in
B_1^{m_\ell + n_\ell}} | k_j|^{-\mu} \Big\| \tW_\ell \big[
 r ; \, k^{(\ell)}_{(m_\ell, n_\ell)} \big] \Big\|_\op
}
\\ \nonumber
& \leq & \frac{1}{\sqrt{p^{(\ell)}_{p_\ell} \, q^{(\ell)}_{q_\ell}}
} \,
\max_j\sup_{r \in I, k^{(\ell)}_{(m_\ell, n_\ell)} \in B_1^{m_\ell +
n_\ell}}| k_j|^{-\mu}\Big\| w_{m_\ell+p_\ell, n_\ell+q_\ell} [ \,
\cdot \, ; \, k^{(\ell)}_{(m_\ell)}, \, \cdot \, ;
 \, \tk^{(\ell)}_{(n_\ell)}, \, \cdot \, \big] \Big\|_{0}
\\ \nonumber
& \le & \frac{1}{\sqrt{p^{(\ell)}_{p_\ell} \, q^{(\ell)}_{q_\ell}} }
\, \Big\| w_{m_\ell+p_\ell, n_\ell+q_\ell}
       \Big\|_{\mu} \period
\end{eqnarray*}
This estimate with $\mu=0$ if $\ell \ne \ell_j$ and $\mu \ge 0$ if
$\ell = \ell_j$, inserted into the $\ell^{th}$ factor on the right
side of (\ref{eq-III-3-14}), yields \eqref{VIII.16} with $s=0$.
%
%

%
To estimate the norm $\big\| V_{\umpnq} \big\|_{\mu,s}$ with $s=1,2$
we need the bounds
\begin{equation} \label{dF}
\big|\partial_r^s F_\ell[r]\big|   \leq   \frac{C_\chi}{\rho}
\end{equation}
where the constant $C_\chi$ is given in \eqref{Cchi}. These bounds
are obtained similarly to \eqref{eq-III-3-5}, using the inequality
\begin{eqnarray*}\big|\partial_r F_\ell[r]\big|  & \; \leq &
 \bigg| \frac{2 \bchi_1[r+\tr_\ell] \, \partial_r
\bchi_1[r+\tr_\ell] }  { T[\rho (r+\tr_\ell)] - E} \bigg| \nonumber
\\ & \; + & \; \bigg| \frac{ \bchi_1^2[r+\tr_\ell] \; \rho \,
\partial_r T[z; \rho (r+\tr_\ell)]} {(T[\rho (r+\tr_\ell)] - E)^2}
\bigg|
\end{eqnarray*}
and a similar inequality for $\big|\partial_r^2 F_\ell[r]\big|$.

To estimate $\big\| V_{\umpnq} \big\|_{\mu,s}$ with $s=1,2$ we apply
the operator $\partial_r^n$
to \eqref{VIII.15} and use the Leibnitz rule of
differentiation of products $s$ times  to obtain
\eqref{VIII.16}.\QED
%

We are now prepared to prove the estimates in
Theorem~\ref{thm-III-2-5}.
Recall that we assume $\rho \le 1/2$
%
and  we choose $\xi = 1/4$. First, we apply Lemma~\ref{lem-VIII.3}
to \eqref{VIII.14} and use that ${m+p \choose p} \leq 2^{m+p}$. This
yields
\begin{eqnarray} \label{VIII.22}
\lefteqn{ \big\| \hw_{M,N} \big\|_{\mu,s} \ \leq \ \sum_{L=1}^\infty
4 \, C_{\chi}\rho^\mu \, L^s \, \Big( \frac{C_{\chi}}{\rho} \Big)^L
\, \big(2 \,\rho \big)^{M+N} \: }
\\ \nonumber
& & \hspace{-7mm}
\sum_{m_1 + \ldots + m_L = M, \atop n_1 + \ldots + n_L = N} \!
\sum_{p_1, q_1, \ldots, p_L, q_L:
      \atop m_\ell + p_\ell + n_\ell + q_\ell \geq 1}
\prod_{\ell = 1}^L \bigg\{ \Big( \frac{ 2 }{ \sqrt{p_\ell} }
\Big)^{p_\ell} \! \Big( \frac{ 2 }{ \sqrt{q_\ell} } \Big)^{q_\ell}
\big\| w_{m_\ell+p_\ell, n_\ell+q_\ell}
       \big\|_{\mu,s} \bigg\} \period
\end{eqnarray}
%
Using the definition \eqref{III.17}, the inequality $2\rho \le 1$,
we derive the following bound for $\hat{\uw}_1:=(\hw_{M,N})_{M+N
\geq 1}$,
\begin{eqnarray} \nonumber
\lefteqn{ \big\| \hat{\uw}_1 \|_{\mu, s, \xi} \ := \ \sum_{M+N \geq
1} \xi^{-(M+N)} \, \big\| \hw_{M,N} \|_{\mu,s} \hspace{30mm} }
\\ \nonumber
& \hspace{-5mm} \leq & 8\, C_{\chi} \, \rho^{1+\mu}
\sum_{L=1}^\infty L^s \, \bigg( \frac{C_{\chi}}{\rho} \bigg)^{L}
\sum_{M+N \geq 1} \sum_{m_1 + \ldots + m_L = M, \atop n_1 + \ldots +
n_L = N} \sum_{p_1, q_1, \ldots, p_L, q_L:
      \atop m_\ell + p_\ell + n_\ell + q_\ell \geq 1}
\\ \nonumber
& \hspace{-5mm} &
\prod_{\ell = 1}^L \bigg\{
\Big( \frac{ 2 \, \xi }{ \sqrt{p_\ell} } \Big)^{p_\ell} \,
\Big( \frac{ 2 \, \xi }{ \sqrt{q_\ell} } \Big)^{q_\ell} \;
\xi^{- ( m_\ell + p_\ell + n_\ell + q_\ell )}\:
\big\| w_{m_\ell+p_\ell, n_\ell+q_\ell}
       \big\|_{\mu,s} \bigg\}
\\ \nonumber
& \hspace{-5mm} \leq & 8\, C_{\chi} \, \rho^{1+\mu}
\sum_{L=1}^\infty L^s \, \bigg( \frac{C_{\chi}}{\rho} \bigg)^{L}
\\ \nonumber
& \hspace{-5mm} & \bigg\{ \sum_{m+n \geq 1} \bigg( \sum_{p=0}^m
\Big( \frac{ 2 \, \xi }{ \sqrt{p} } \Big)^{p} \bigg) \, \bigg(
\sum_{q=0}^n \Big( \frac{ 2 \, \xi }{ \sqrt{q} } \Big)^{q} \bigg) \,
\xi^{- (m+n)}\: \| w_{m,n} \|_{\mu, s} \bigg\}^L .
\end{eqnarray}
Using the assumption $\xi
= 1/4$ and the estimate $\sum_{p=0}^m ( 2 \, \xi/\sqrt{p} )^{p} \le
\sum_{p=0}^\infty \big( 2 \, \xi \big)^{p} \ =  \frac{1}{1 \: - \: 2
\, \xi }$, and recalling the definitions $\uw_1 := (w_{m,n})_{m+n
\geq 1}$ and $\big\| \uw_1 \|_{\mu, s, \xi} \ := \\$ $\sum_{M+N \geq
1} \xi^{-(m+n)} \, \big\| w_{m,n} \|_{\mu,s} $, we obtain
\begin{eqnarray}
\label{eq-III-3-18}
\big\| \hat{\uw}_1 \|_{\mu, s, \xi} \leq & 8\C \, \rho^{\mu+1}
\sum_{L=1}^\infty L^s \, B^L,
\end{eqnarray}
where
\begin{equation} \label{VI.25}
B:= \frac{C_{\chi} }{\rho(1 \: - \: 2 \, \xi)^2} \, \big\| \uw_1
\big\|_{\mu, s,\xi} .
\end{equation}

%
%

Note that in (\ref{eq-III-3-18}) we have dropped the factor
$p^{-p/2}$ gained in Theorem~\ref{thm-III.1}. Our assumption,
$\gamma \leq  (8 \, C_{\chi})^{-1} \rho$,  also insures that
\begin{equation} \label{eq-III-3-20}
B \
\leq \ \frac{4 \, C_{\chi} \, \gamma}{\rho} \ \leq \ \frac{1}{2} .
\end{equation}
Thus the geometric series in the last line of (\ref{eq-III-3-18}) is
convergent. We obtain for $s=0, 1, 2$
\begin{equation} \label{eq-III-3-21}
\sum_{L=1}^\infty L^s \, B^L \
\leq \ 8 \, B \period
\end{equation}
Inserting (\ref{eq-III-3-21}) into (\ref{eq-III-3-18}), we see that
the r.h.s. of \eqref{eq-III-3-18} is bounded by $64 \, C_{\chi} \,
\rho^{1+\mu} \, B $ which, remembering the definition of $B$ gives
\begin{equation} \label{est-W}
\big\| \hat{\uw}_1 \|_{\mu, s,\xi} \leq 256 \, C_{\chi}^2 \,
\rho^{\mu} \: \big\| \uw_1 \big\|_{\mu,s, \xi} \period
\end{equation}
%

Next, we estimate $\hw_{0,0}$. We analyze the expression
\eqref{VI.15}.
%
%
%
%
%
%
%
%
Using estimate Eq.~(\ref{VIII.16}) with $\underline{m}=0,
\underline{n}=0$ (and consequently, $M=0, N=0$), we find
\begin{equation} \label{eq-III-3-24}
\rho^{- 1} \: \big\| V_{\upq} \|_{\mu,s} \ \leq \ 2 L^s \,
C_{\chi}^{L+1} \, \rho^{- L} \, \prod_{\ell = 1}^L \frac{ \big\|
w_{p_\ell, q_\ell}
       \big\|_{\mu,s} }
     { \sqrt{p_\ell^{p_\ell} \, q_\ell^{q_\ell}} } .
\end{equation}
In fact, examining the proof of Lemma \ref{lem-VIII.3}  more
carefully we see that the following, slightly stronger estimate is
true
\begin{equation} \label{eq-III-3-24a}
\rho^{- 1} \: \big\| \partial_r^sV_{\upq} \|_{\mu,0} \ \leq \ 2 L^s
\, C_{\chi}^{L+1} \, \rho^{- L+s} \, \prod_{\ell = 1}^L \frac{
\big\|
w_{p_\ell, q_\ell}
       \big\|_{\mu,s} }
     { \sqrt{p_\ell^{p_\ell} \, q_\ell^{q_\ell}} } .
\end{equation}
Now, using
\eqref{eq-III-3-24a} and $\sum_{p+q \geq 1} \big\| w_{p,q}
\big\|_{\mu,s}  \le \xi \sum_{p+q \geq 1}\xi^{-p-q} \big\| w_{p,q}
\big\|_{\mu,s}  =: \|
\uw_1 \|_{\mu,s, \xi}$, where, recall, $\uw_1:=( w_{m,n} )_{m+n \geq
1}$,
we obtain
\begin{eqnarray} \nonumber
&& \: \rho^{-1} \sum_{L=2}^\infty \sum_{p_1, q_1, \ldots, p_L, q_L:
\atop p_\ell + q_\ell \geq 1} \sup_{r \in I} \big|
\partial_r^s V_{\upq} [ r ] \big|
\\ \nonumber
& \leq &  \: 2 \, C_{\chi} \rho^s\, \sum_{L=2}^\infty L^s \, \Big(
\frac{C_{\chi}}{\rho} \Big)^L \, \bigg\{ \sum_{p+q \geq 1} \big\|
w_{p,q} \big\|_{\mu,s} \bigg\}^L\\ 
& \leq &  \: 2 \, C_{\chi} \rho^s\, \sum_{L=2}^\infty L^s \, D^L,
\end{eqnarray}
where $D := C_{\chi}\rho^{-1}  \xi \|
\uw_1 \|_{\mu,s, \xi}$. Now, since $D \leq C_{\chi} \xi \rho^{-1}
\gamma \le \xi/8 = 1/16$, we have, similarly to (\ref{eq-III-3-21}),
that $\sum_{L=2}^\infty L^s D^L \leq 12 D^2$ for $s=0,1,2$. Hence
we find
\begin{eqnarray} \nonumber
& \: \rho^{-1} & \sum_{L=2}^\infty \sum_{p_1, q_1, \ldots, p_L, q_L:
\atop p_\ell + q_\ell \geq 1} \sup_{r \in I} \big|
\partial_r^s V_{\upq} [ r ] \big|
\\ \label{eq-III-3-24b}
& \leq &  \: 24 \, C_{\chi} \rho^s\, \Big( \frac{C_{\chi} \,
\xi}{\rho} \, \big\| \uw_1 \big\|_{\mu,s, \xi} \Big)^2,
\end{eqnarray}
for $s= 0, 1, 2$.

We set $\hE := \hw_{0,0}[ 0]$. Since $E = w_{0,0}[ 0]$,
Eqs.~\eqref{VI.15} and \eqref{eq-III-3-24b} yield
\begin{equation} \label{est-E}
\big| \hE \, - \, \rho^{-1} E \big| \ \leq \ 24 \, C_{\chi} \, \Big(
\frac{C_{\chi} \, \xi}{\rho} \, \big\| \uw_1 \big\|_{\mu,0, \xi}
\Big)^2.
\end{equation}
Next, writing $\hT[ r] := \hw_{0,0}[  r ] - \hw_{0,0}[  0 ]$, we
find furthermore that
\begin{eqnarray} \nonumber
\lefteqn{ \sup_{r \in [0,\infty)}\big| \hT'[ r] - 1 \big| \ \; = \
\; \sup_{r \in [0,\infty)} \big| \partial_r \hw_{0,0}[  \, r ] \, -
\, 1 \big| }
\\  \label{est-T'}
& \leq & \sup_{r \in [0,\infty)}\big| T'[ r] - 1 \big| \: + \: 24 \,
C_{\chi} \rho\, \Big( \frac{C_{\chi} \, \xi}{\rho} \, \big\| \uw_1
\big\|_{\mu,1, \xi} \Big)^2.
\end{eqnarray}

Now, recall that $\big| T'[ r] - 1 \big| \le \beta$ and $\big\|
\uw_1 \big\|_{\mu,s, \xi}\le \gamma$. Hence Eqns \eqref{est-E},
\eqref{est-T'},
and \eqref{est-W} give \eqref{eqn:23} with $\alpha' = 24 \, C_{\chi}
\, \Big( \frac{C_{\chi} \, \xi\gamma }{\rho} \, \Big)^2$, $\beta'=
\beta  + 24C_{\chi} \, \Big( \frac{C_{\chi} \, \xi\gamma}{\rho}
\Big)^2$ and $\gamma'= 256 \, C_{\chi}^2 \, \rho^{\mu}\gamma$.
Remembering that $\xi = \sqrt{\rho} / (4 C_{\chi})$ we conclude that
the statement of Theorem~\ref{thm-III-2-5} holds.\QED

\begin{remark} \label{rem-VI.5}
In the proof the limiting absorption principle (LAP) in
\cite{FroehlichGriesemerSigal2008b} to estimate $\big\| V_{\umpnq}
\big\|_{\mu,s}$, with $s=1,2$, (see Lemma \ref{lem-VIII.3}) instead
of the operator $\partial_r^n$, we apply the operator  $\partial_r^n
(k\partial_k)^q$ to \eqref{VIII.15}. Here $q:= (q_1,  \ldots,
q_{M+N}),$ $ (k\partial_k)^q: = \prod_1^{M+N}(k_j \cdot
\nabla_{k_j})^{q_j}$, with $k_{m+j} := \tk_j$, and
the indices $n$ and $q$ satisfy $0 \le n+|q| \leq s$.
\end{remark}

\begin{remark} \label{rem-VI.6}
For the proof of the limiting absorption principle in
\cite{FroehlichGriesemerSigal2008b} we also need the following
estimate (here we use that $\hT''[ r] = \partial_r ^2 \hw_{0,0}[  \,
r ]$)
\begin{equation} \label{est-T''}
\sup_{r \in [0,\infty)}\big| \hT''[ r]  \big| \ \;
 \leq  \rho \sup_{r \in [0,\infty)}\big| T''[ r]  \big| \: + \: 24
\, C_{\chi} \rho^2\, \Big( \frac{C_{\chi} \, \xi}{\rho} \, \big\|
\uw_1 \big\|_{\mu,2, \xi} \Big)^2.
\end{equation}

\end{remark}

\secct{Appendix II: Construction of Eigenvalues and Eigenvectors}
\label{sec-VII}
%
In this appendix we prove that the value $E:=e(H_s)+H_u$ we
constructed in Section V is the ground state energy of the
Hamiltonian $H$ under consideration (see Theorem V.3) and we
construct the corresponding the ground state energy. We use the
definitions of Section V.


%
%


%
%
%
%
%
\begin{theorem}
\label{thm-XI.1}
Let $H \in \cD$
Then the value $E:=e(H_s)+H_u$ where $e(H_s)$ is given in Theorem
V.1,
is a simple eigenvalue of the operator $H$. The corresponding
eigenfunction is given constructively in Eq.~(\ref{eq-III-4-25.7})
below.
%
%
\end{theorem}
\Proof Let  $H^{(0)} := H-E \in \cM_s$. We define a sequence of
operators $( H^{(n)} )_{n = 0}^\infty$ in $\cW_{op}^{\mu,s}
\subseteq \cB(\cH_\red)$ by $H^{(n)}:= \cR_\rho^n \big( H^{(0)}
\big)$. We will also need the following representation for $S_\rho$:
\begin{equation} \label{eq-III-2-1}
S_\rho(A) \ =: \
\Gamma_\rho \; A \; \Gamma_\rho^*
\comma
\end{equation}
where $\Gamma_\rho$ is the unitary dilatation on $\cF$ defined by
this formula and $\Gamma_\rho \Om = \Om $.
%
Then the definition (\ref{eqn:23}) of $\cR_\rho$ implies that, for
all integers $n \geq 0$,
%
\begin{eqnarray} \label{eq-III-4-22}
H^{(n)} =  \frac{1}{\rho} \, \Gamma_\rho \: \Big( F_{\rho} \big( \,
H^{(n-1)} \
 \, \big) \Big)\:
\Gamma_\rho^*,
\end{eqnarray}
%
%
where, recall, $F_{\rho} := F_{\tau \chi_\rho}$ with
$\tau(H):=W_{0,0}$ (see Eqn (IV.1)).
We will use the operators $Q_{\tau \chi}$ defined in
\eqref{eq-II-4}. It is easy to show (see
\cite{BachChenFroehlichSigal2003}) that these operators satisfy the
identity $HQ_{\tau \chi}=\chi F_{\tau \chi}(H)$. Let
\begin{equation}\\ \label{eq-III-4-22.1b} Q^{(n)}  :=
Q_{\tau \chi_\rho} \big( \, H^{(n)}  \
 \, \big) \comma
\hspace{9mm}
\end{equation}
Then the equation $H^{(n)}Q^{(n)}=\chi_\rho F_{\rho}(H^{(n)})$
%
%
together with (\ref{eq-III-4-22}), implies the
intertwining property
\begin{equation} \label{eq-III-4-24}
H^{(n-1)} \; Q^{(n-1)} \; \Gamma_\rho^* \ = \ \rho \; \Gamma_\rho^*
\; \chi_1 \; H^{(n)} \period
\end{equation}
Eq.~(\ref{eq-III-4-24}) is the \emph{key identity} for the proof of
the \emph{existence of an eigenvector} with the eigenvalue $e$.

For the construction of this eigenvector, for non-negative integers
$ \beta $ we define vectors $\Psi_{k}$ in $\cH_\red$ by setting
$\Psi_{0} := \Om$ and
\begin{equation} \label{eq-III-4-25}
\Psi_{k} \ := \ Q^{(0)} \; \Gamma_\rho^* \; Q^{(1)} \; \Gamma_\rho^*
\cdots Q^{(k-1)} \; \Om \period
\end{equation}
We first show that this sequence is convergent, as $k \to \infty$.
To this end, we observe that $\Om = \Gamma_\rho^* \, \chi_\rho \,
\Om$ and hence
\begin{equation} \label{eq-III-4-25.1}
\Psi_{k+1} - \Psi_{k} \ = \ Q^{(0)} \; \Gamma_\rho^* \; Q^{(1)} \;
\Gamma_\rho^* \cdots Q^{(k-1)} \; \Gamma_\rho^* \, \big( Q^{(k)} -
\chi_\rho \big) \, \Om \period
\end{equation}
Since $\|\chi_\rho\| \leq 1$, this implies that
\begin{equation} \label{eq-III-4-25.2}
 \big\| \Psi_{k+1} - \Psi_{k} \big\|
\ \leq \ \big\| Q^{(k)} - \chi_\rho \big\|_\op \:
\prod_{j=0}^{\beta-1} \Big\{ 1 \, + \, \big\| Q^{(j)} - \chi_\rho
\big\|_\op \Big\} \period
\end{equation}

To estimate the terms on the r.h.s. we consider the $j$-th step
Hamiltonian $H^{(j)}$. As in the proof of proposition VI.5 we write
$H^{(j)}$ as
%
\begin{equation} \label{eq-III-4-21}
H^{(j)} \ = \ E_j \, \cdot \one \, + T_{j} \, + \, W_{j}
  \ \comma
\end{equation}
with
%
%
%
\begin{equation} \label{eq-III-4-25.4}
| E_j| \le 8 \alpha_j\ \mbox{and}\    \| W_{j} \|_\op \ \leq \
\gamma_j \ \leq \ \frac{\rho}{16} \period
\end{equation}
Recalling the definition (\ref{eq-II-4}) of $Q_{(j)}$, we have
\begin{eqnarray} \label{eq-III-4-25.3}
\lefteqn{ \chi_\rho - Q_{(j)} \ = \ }
\\ \nonumber & & \hspace{-5mm}
\bchi_\rho \big(E_{j}  + T_{j} \, + \, \bchi_\rho \, W_{j} \,
\bchi_\rho \big)^{-1} \bchi_\rho \, W_{j} \, \chi_\rho \period
\end{eqnarray}
%
By \eqref{eq-III-4-25.4}, for all $j \in \NN$, we may estimate
\begin{equation} \label{eq-III-4-25.5}
\| \chi_\rho - Q_{(j)} \|_\op \ \leq \ \Big( \frac{\rho}{8} \, - \,
\| W_{j} \|_\op \Big)^{-1} \,  \| W_{j} \|_\op \ \leq \ \frac{16 \,
\gamma_j}{\rho} \period
\end{equation}
Inserting this estimate into (\ref{eq-III-4-25.2}) and using that
$\prod_{j=0}^\infty (1 + \lambda_j) \leq \exp\Big[\sum_{j=0}^\infty
\lambda_j \Big]$, for $\lambda_j \geq 0$, we obtain
\begin{eqnarray} \label{eq-III-4-25.6}
\big\| \Psi_{k+1} - \Psi_{k} \big\| & \leq & \frac{16 \,
\gamma_{k}}{\rho} \: \prod_{j=0}^{k-1} \Big\{ 1 \, + \, \frac{16 \,
\gamma_j}{\rho}  \Big\}
\nonumber \\
& \leq & \frac{16 \, \gamma_{k}}{\rho} \, \exp\big[
 32 \, \gamma_0 \, \rho^{-1} \big] ,
\end{eqnarray}
where we have used that $\sum_{j=0}^\infty \gamma_j \le 2\gamma_0$
(recall the definition of $\gamma_{j}$ after Eqn (\ref{rho})). Since
$\sum_{j=0}^\infty \gamma_j < \infty$, we see that the sequence $(
\Psi_{k} )_{k \in \NN_0}$ of vectors in $\cH_\red$ is convergent,
and its limit
\begin{equation} \label{eq-III-4-25.7}
\Psi_{\infty} \ := \ \lim_{k \to \infty} \Psi_{k} \comma
\end{equation}
satisfies the estimate
\begin{equation} \label{eq-III-4-25.8}
\big\| \Psi_{ \infty} - \Om \big\| \ = \ \big\| \Psi_{ \infty} -
\Psi_{ 0} \big\| \ \leq \
 \frac{32 \, \gamma_0}{\rho} \, \exp\big[  32 \,
\gamma_0 \, \rho^{-1} \big] \comma
\end{equation}
which guarantee that $\Psi_{( \infty)} \neq 0$.

The vector $\Psi_{ \infty}$ constructed above is an element of the
kernel of $H^{(0)}$, as we will now demonstrate. Observe that,
thanks to (\ref{eq-III-4-24}),
\begin{eqnarray} \label{eq-III-4-26}
H^{(0)} \, \Psi_{k} & = & \big( H^{(0)} \; Q^{(0)} \; \Gamma_\rho^*
\big) \, \big( Q^{(1)} \; \Gamma_\rho^* \cdots Q^{(k-1)} \; \Om
\big)
\nonumber \\
& = & \rho \, \Gamma_\rho^* \, \chi_1\, \big( H^{(1)} \; Q^{(1)} \;
\Gamma_\rho^* \big) \big( Q^{(2)} \; \Gamma_\rho^* \cdots Q^{(k-1)}
\; \Om \big)
\nonumber \\
& \vdots &
\nonumber \\
& = & \rho^{k } \, \big( \Gamma_\rho^* \; \chi_1 \big)^{k } \;
H^{(k)} \: \Om \period
\end{eqnarray}
%

%
%
Eq \eqref{eq-III-4-21} together with the estimate
(\ref{eq-III-4-25.4}) and the relation $T_{k}\Om =0$ implies that
\begin{eqnarray} \label{eq-III-4-27}
\big\|  \: H^{(k)} \: \Om \big\| & = & \big\|  \: (W_{k} \,  + \;
E_{k}) \, \Om \big\|
\\ \nonumber
& \leq & \gamma_{k}  + \: 8\alpha_{k}^2 \ \leq \ 2\gamma_{k} \,
\period
\end{eqnarray}
%
Summarizing (\ref{eq-III-4-26})--(\ref{eq-III-4-27}) and using that
the operator norm of $\Gamma_\rho^* \: \chi_1$ is bounded by $1$, we
arrive at
\begin{equation} \label{eq-III-4-28}
\big\| H^{(0)} \, \Psi_{k} \big\| \ \leq \ 2\gamma_{k} \ \to \ 0
\end{equation}
as $k \to \infty$.  Since $H^{(0)} \in \cB(\cH_\red)$ is continuous,
(\ref{eq-III-4-28}) implies that
\begin{equation} \label{eq-III-4-28.1}
H^{(0)} \, \Psi_{ \infty} \ = \ \lim_{k \to \infty} \big( H^{(0)} \,
\Psi_{k} \big) \ = \ 0 .
\end{equation}

Thus $0$ is an eigenvalue of the operator $H^{(0)}:=H-E$, i.e. $E$
is an eigenvalue of the operator $H$, with the eigenfunction $\Psi_{
\infty}$.
\QED

%
%

\secct{Appendix III: Analyticity of all Parts of $H(\underline{w})$}
\label{sect-analytic}

Let $S$ be an open set in a Banach space $\cB$. Below the
analyticity is understood in the sense described in the paragraph
preceding Theorem \ref{stable-manif}.

\begin{proposition}[\cite{GriesemerHasler2}] \label{analytic00}
Suppose that $\lambda\mapsto H(\underline{w}^{\lambda}) 
$ is analytic in $\lambda\in S 
$ and that $H(\underline{w}^{\lambda})$ belongs to some polydisc
$\cD(\alpha,\beta,\gamma)$ for all $\lambda\in S$. Then:
$$
     \lambda\mapsto w_{0,0}^{\lambda}(H_f)\ \quad  
\mbox{and}\ \quad \lambda\mapsto W(\uw^{\lambda}) $$
are analytic in $\lambda\in S$.
\end{proposition}

\begin{proof}
Recall that $B_1=\{k\in\RR^3:|k|\leq 1\}$ and that an operator $A$ is called $H_f$-bounded
iff the operator $A(H_f+1)^{-1}$ is bounded.
Let $P_1$ denote the projection onto the one boson subspace of
$\mathcal{F}$, which is isomorphic to $L^2(\RR^3)$. Then
$P_1 H(\uw^{\lambda}) P_1$, like $H(\uw^{\lambda})$, is analytic. We write
\begin{eqnarray}
P_1 H(\uw^{\lambda}) (H_f+1)^{-1} P_1 &=& P_1 w_{0,0}^{\lambda}(H_f)(H_f+1)^{-1} P_1 +
P_1 W_{1,1}(\underline{w}^{\lambda})(H_f+1)^{-1} P_1 \nonumber \\
&=& D_{\lambda} + K_{\lambda}  \label{eq:h00star} \; ,
\end{eqnarray}
where $D_{\lambda}$ denotes multiplication with $w^{\lambda}_{0,0}(\omega)(\omega+1)^{-1},\ \omega:=|k|,$
and $K_{\lambda}$ is the Hilbert Schmidt operator with kernel
$$
M_{\lambda}(k, \tilde{k} ) =  w^{\lambda}_{1,1}(0,k, \tilde{k} )(\tilde{\omega}+1)^{-1},
$$
whose support belongs to $B_1\times B_1$. In what follows if an operator family has a factor $(H_f+1)^{-1} $ standing on its right, then the analyticity is understood in then operator norm.

Our strategy is to show first that $K_{\lambda}$ and hence $$P_1
w_{0,0}^{\lambda}(H_f)(H_f+1)^{-1} P_1 = P_1 H(w^{\lambda})(H_f+1)^{-1}  P_1 - K_{\lambda}$$ are
analytic. Then we show that $\lambda \mapsto w_{0,0}^{\lambda}(H_f)$ is analytic.
The analyticity of $\lambda\mapsto W(\uw^{\lambda})=H(\uw^{\lambda}) -w_{0,0}^{\lambda}(H_f)$ then follows.\\

\noindent\underline{Step 1}: $K_{\lambda}$ is analytic.

For each $n \in \NN$ let $\{ Q_i^{(n)}\}_i$ be a collection of $n$
measurable subsets of $B_1$ such that
\begin{equation} \label{eq:h001}
     B_1 = \bigcup_{i=1}^n Q_i^{(n)} \ , \quad Q^{(n)}_i \cap Q^{(n)}_j =
     \emptyset, \ \ i \neq j \; ,
\end{equation}
and
\begin{equation} \label{eq:h002}
   |Q^{(n)}_i | \leq \frac{\rm const}{n} \; .
\end{equation}
Let $\chi_i^{(n)}$ denote the operator on $L^2(B_1)$ of multiplication
with $\chi_{Q^{(n)}_i}$. Then for $i\neq j$, $\chi_i^{(n)}
D_{\lambda} \chi_j^{(n)} = 0$ because $\chi_i^{(n)}$ and
$\chi_j^{(n)}$ have disjoint support and commute with $D_{\lambda}$.
Together with (\ref{eq:h00star}) this implies that
$$
\chi^{(n)}_i K_{\lambda} \chi^{(n)}_j = \chi^{(n)}_i P_1
H(\underline{w}^{\lambda})(H_f+1)^{-1}  P_1 \chi_j^{(n)} \ , \quad {\rm for} \ \ i \neq j.
$$
Since the right hand side is analytic, so is the left hand side
and hence
$$
K^{(n)}_{\lambda}  = \sum_{i \neq j} \chi^{(n)}_i K_{\lambda}
\chi^{(n)}_j \;
$$
is analytic. It follows that $\lambda \mapsto \sprod{\varphi}{
K^{(n)}_{\lambda}\psi}$ is analytic for all $\varphi, \psi$ in
$L^2(B_1)$. Now let $\varphi, \psi \in C(B_1)$. Then
\begin{eqnarray*}
\lefteqn{ \left| \sprod{\varphi}{K^{(n)}_{\lambda} \psi } -
\sprod{\varphi}{ K_{\lambda} \psi }
\right| } \\
&= & \left|\int_{B_1\times B_1} \overline{\varphi}(x) \psi(y)
M_{\lambda}(x,y)\sum_{i=1}^n\chi^{(n)}_i(x) \chi^{(n)}_i(y) dx dy\right| \\
&\leq & \| \varphi \|_\infty \|\psi \|_\infty \| K_{\lambda} \|_{\rm
HS} \left(\sum_{i=1}^n |Q_i^{(n)}|^2\right)^{1/2} \longrightarrow 0
\; , \quad ( n \to \infty ),
\end{eqnarray*}
uniformly in $\lambda$, because the Hilbert Schmidt norm
$\|K_{\lambda} \|_{\rm HS}$ is bounded uniformly in $\lambda$ (in
fact, it is bounded by $\gamma$). This proves that $\sprod{\varphi}{
K_{\lambda}\psi}$ is analytic for all $\varphi, \psi \in C(B_1)$.
Since $C(B_1)$ is dense in $L^2(B_1)$, an other approximate argument
using $\sup_{\lambda} \| K_{\lambda} \|<\infty$ shows that
$\sprod{\varphi}{ K_{\lambda} \psi}$ is analytic for
all $\varphi, \psi \in L^2(B_1)$. Therefore $\lambda\mapsto K_{\lambda}$ is analytic \cite{Kato}.\\


\noindent \underline{Step 2}: For each $k \in\RR^3$,
$w^{\lambda}_{0,0}(|k|)(\omega+1)^{-1} $ is an analytic function of $\lambda$.

For each $n\in\NN$ let $f_{k,n}\in L^2(B_1)$ denote a multiple of the
characteristic function of $B_{1/n}(k)$ with $\|f_{k,n}\|
= 1$. By the continuity of $w^{\lambda}_{0,0}(|k|)$ as a function of $k$
\begin{eqnarray}
w^{\lambda}_{0,0}(|k|)(\omega+1)^{-1}  &=& \lim_{n \to \infty} \int_{\RR^3} | f_{k,n}(x) |^2
w^{\lambda}_{0,0}(|x|) (|x|+1)^{-1} dx\label{eq:h001a} \\
&=& \lim_{n \to \infty} \langle a^*(f_{k,n}) \Omega, w_{0,0}^{\lambda}(H_f)(H_f+1)^{-1}
a^*(f_{k,n}) \Omega \rangle.\nonumber
\end{eqnarray}
Since $a^*(f_{k,n}) \Omega \in P_1\cF$ the
expression $\langle\cdots\rangle$, before taking the limit, is an
analytic function of $\lambda$. By assumption on $w^{\lambda}_{0,0}$,
this function is Lipschitz continuous with respect to $|k|$
\emph{uniformly in $\lambda$}. Therefore the convergence in
(\ref{eq:h001a}) is uniform in $\lambda$ and hence
$w^{\lambda}_{0,0}(|k|)(\omega+1)^{-1} $ is analytic by the Weierstrass
approximation theorem from complex analysis.\\

\noindent\underline{Step 3}: $w^{\lambda}_{0,0}(H_f)$ is analytic.

By the spectral theorem
$$
\langle \varphi, w^{\lambda}_{0,0}(H_f)(H_f+1)^{-1} \varphi \rangle =
\int_{[0,\infty)} w^{\lambda}_{0,0}(x)(x+1)^{-1} d\mu_{\varphi}(x).
$$
By an application of Lebesgue's dominated convergence theorem, using
$\sup_{\lambda} \|w^{\lambda}_{0,0}(x+1)^{-1}\| < \infty$, we see that the
right hand side, which we call  $\varphi(\lambda)$, is a continuous
function of $\lambda$. Therefore
$$
\int_\Gamma \varphi(\lambda) d\lambda =
\int_{[0,1]}\left(\int_{\Gamma}w^{\lambda}_{0,0}(x)(x+1)^{-1}\,d\lambda\right) d\mu_{\varphi} (x)
$$
for all closed loops $\Gamma:t\mapsto \lambda(t)$ in $S$. The analyticity of $\lambda\mapsto\vphi(\lambda)$ now follows from the analyticity of
$w^{\lambda}_{0,0}(x)(x+1)^{-1} $ and the theorems of Cauchy and Morera.
By polarization, $w^{\lambda}_{0,0}(H_f)(H_f+1)^{-1} $ is weakly analytic and hence
analytic.

\end{proof}

\secct{Supplement: Background on the Fock space, etc}
\label{sect-SA}
%
%
Let $ \fh$ be either $ L^2 (\RR^3, \mathbb{C}, d^3 k)$ or  $ L^2
(\RR^3, \mathbb{C}^2, d^3 k)$. In the first case we consider $ \fh$
as the Hilbert space of one-particle states of a scalar Boson or a
phonon, and in the second case,  of a photon. The variable
$k\in\RR^3$ is the wave vector or momentum of the particle. (Recall
that throughout this paper, the velocity of light, $c$, and Planck's
constant, $\hbar$, are set equal to 1.) The Bosonic Fock space,
$\cF$, over $\fh$ is defined by
\begin{equation} \label{eq-I.10}
\cF \ := \ \bigoplus_{n=0}^{\infty} \cS_n \, \fh^{\otimes n} \comma
\end{equation}
where $\cS_n$ is the orthogonal projection onto the subspace of
totally symmetric $n$-particle wave functions contained in the
$n$-fold tensor product $\fh^{\otimes n}$ of $\fh$; and $\cS_0
\fh^{\otimes 0} := \CC $. The vector $\Om:=1
\bigoplus_{n=1}^{\infty}0$ is called the \emph{vacuum vector} in
$\cF$. Vectors $\Psi\in \cF$ can be identified with sequences
$(\psi_n)^{\infty}_{n=0}$ of $n$-particle wave functions,  which are
totally symmetric in their $n$ arguments, and $\psi_0\in\CC$. In the
first case these functions are of the form, $\psi_n(k_1, \ldots,
k_n)$, while in the second case, of the form $\psi_n(k_1, \lambda_1,
\ldots, k_n, \lambda_n)$, where $\lambda_j \in \{-1, 1\}$ are the
polarization variables.

In what follows we present some key definitions in the first case
only limiting ourselves to remarks at the end of this appendix on
how these definitions have to be modified for the second case. The
scalar product of two vectors $\Psi$ and $\Phi$ is given by
\begin{equation} \label{eq-I.11}
\la \Psi \, , \; \Phi \ra \ := \ \sum_{n=0}^{\infty}  \int
\prod^n_{j=1} d^3k_j \; \overline{\psi_n (k_1, \ldots, k_n)} \:
\vphi_n (k_1, \ldots, k_n) \period
\end{equation}

Given a one particle dispersion relation $\om(k)$, the energy of a
configuration of $n$ \emph{non-interacting} field particles with
wave vectors $k_1, \ldots,k_n$ is given by $\sum^{n}_{j=1}
\om(k_j)$. We define the \emph{free-field Hamiltonian}, $\hf$,
giving the field dynamics, by
%
\begin{equation} \label{eq-I.17a}
(\hf \Psi)_n(k_1,\ldots,k_n) \ = \ \Big( \sum_{j=1}^n \om(k_j) \Big)
\: \psi_n (k_1, \ldots, k_n) ,
\end{equation}
for $n\ge1$ and $(\hf \Psi)_n =0$ for $n=0$. Here
$\Psi=(\psi_n)_{n=0}^{\infty}$ (to be sure that the r.h.s. makes
sense we can assume that $\psi_n=0$, except for finitely many $n$,
for which $\psi_n(k_1,\ldots,k_n)$ decrease rapidly at infinity).
Clearly that the operator  $\hf$ has the single eigenvalue  $0$ with
the eigenvector $\Om$ and the rest of the spectrum absolutely
continuous.

With each function $\vphi \in \fh$ one associates an
\emph{annihilation operator} $a(\vphi)$ defined as follows. For
$\Psi=(\psi_n)^{\infty}_{n=0}\in \cF$ with the property that
$\psi_n=0$, for all but finitely many $n$, the vector $a(\vphi)
\Psi$ is defined  by
\begin{equation} \label{eq-I.12}
(a(\vphi) \Psi)_n (k_1, \ldots, k_n) \ := \ \sqrt{n+1 \,} \, \int
d^3 k \; \overline{\vphi(k)} \: \psi_{n+1}(k, k_1, \ldots, k_n).
\end{equation}
These equations define a closable operator $a(\vphi)$ whose closure
is also denoted by $a(\vphi)$. Eqn \eqref{eq-I.12} implies the
relation
\begin{equation} \label{eq-I.13}
a(\vphi) \Om \ = \ 0 \period
\end{equation}
The creation operator $a^*(\vphi)$ is defined to be the adjoint of
$a(\vphi)$ with respect to the scalar product defined in
Eq.~(\ref{eq-I.11}). Since $a(\vphi)$ is anti-linear, and
$a^*(\vphi)$ is linear in $\vphi$, we write formally
\begin{equation} \label{eq-I.14}
a(\vphi) \ = \ \int d^3 k \; \overline{\vphi(k)} \, a(k) \comma
\hspace{8mm} a^*(\vphi) \ = \ \int d^3 k \; \vphi(k) \, a^*(k)
\comma
\end{equation}
where $a(k)$ and $a^*(k)$ are unbounded, operator-valued
distributions. The latter are well-known to obey the \emph{canonical
commutation relations} (CCR):
\begin{equation} \label{eq-I.15}
\big[ a^{\#}(k) \, , \, a^{\#}(k') \big] \ = \ 0 \comma \hspace{8mm}
\big[ a(k) \, , \, a^*(k') \big] \ = \ \delta^3 (k-k') \comma
\end{equation}
where $a^{\#}= a$ or $a^*$.

Now, using this one can rewrite the quantum Hamiltonian $\hf$ in
terms of the creation and annihilation operators, $a$ and $a^*$, as
\begin{equation} \label{Hfa}
\hf \ = \ \int d^3 k \; a^*(k)\; \om(k) \; a(k) \comma
\end{equation}
acting on the Fock space $ \cF$.

More generally, for any operator, $t$, on the one-particle space $
\fh$ we define the operator $T$ on the Fock space $\cF$ by the
following formal expression $T: = \int a^*(k) t a(k) dk$, where the
operator $t$ acts on the $k-$variable ($T$ is the second
quantization of $t$). The precise meaning of the latter expression
can obtained by using a basis $\{\phi_j\}$ in the space $ \fh$ to
rewrite it as $T: = \sum_{j} \int a^*(\phi_j) a(t^* \phi_j) dk$.

To modify the above definitions to the case of photons, one replaces
the variable $k$ by the pair $(k, \lambda)$ and adds to the
integrals in $k$ also the sums over $\lambda$. In particular, the
creation and annihilation operators have now two variables: $a_
\lambda^\#(k)\equiv a^\#(k, \lambda)$; they satisfy the commutation
relations
\begin{equation} \label{eq-I.16}
\big[ a_{\lambda}^{\#}(k) \, , \, a_{\lambda'}^{\#}(k') \big] \ = \
0 \comma \hspace{8mm} \big[ a_{\lambda}(k) \, , \,
a_{\lambda'}^*(k') \big] \ = \ \delta_{\lambda, \lambda'} \delta^3
(k-k') .
\end{equation}
One can also introduce the operator-valued transverse vector fields
by
$$a^\#(k):= \sum_{\lambda \in \{-1, 1\}} e_{\lambda}(k) a_{\lambda}^\#(k),$$
where $e_{\lambda}(k) \equiv e(k, \lambda)$ are polarization
vectors, i.e. orthonormal vectors in $\mathbb{R}^3$ satisfying $k
\cdot e_{\lambda}(k) =0$. Then in order to reinterpret the
expressions in this paper for the vector (photon) - case one either
adds the variable $\lambda$ as was mentioned above or replaces, in
appropriate places, the usual product of scalar functions or scalar
functions and scalar operators by the dot product of
vector-functions or vector-functions and operator valued
vector-functions.



\vspace{3mm}   \noindent {\bf Acknowledgements:}

A part of this work was done while the third author was visiting ETH
Z\"urich, ESI Vienna and IAS Princeton. He is grateful to these
institutions for hospitality.
 \vspace{3mm}



\end{document}